\documentclass[12pt,a4paper]{article}
\usepackage{times}
\usepackage{fullpage}
\usepackage{graphicx}
\usepackage{amsthm}
\usepackage{amsmath}
\usepackage{amsfonts}
\usepackage{amssymb}
\usepackage{dsfont}
\usepackage{setspace}
\usepackage{algorithm}
\usepackage{algorithmic}
\usepackage{url}
\usepackage{subfigure}
\usepackage{comment}
\usepackage{paralist}
\usepackage{multirow}
\IfFileExists{srcltx.sty} {\usepackage[active]{srcltx}}

\DeclareMathOperator*{\argmin}{arg\,min}

\newcommand{\girth}{\mathrm{girth}}%
\newcommand{\eqdf }{\triangleq}%
\newcommand{\wone }{\lVert w \rVert_1}%
\newcommand{\xone }{\lVert x \rVert_1}%
\newcommand{\R }{\mathds{R}}%notation for Reals
\newcommand{\E }{\mathds{E}}%notation for Expectation
\newcommand{\N }{\mathds{N}}
\newcommand{\calN }{\mathcal{N}} %neighborhood
\newcommand{\calC }{\mathcal{C}} %code
\newcommand{\calT }{\mathcal{T}} %trees
\newcommand{\calV }{\mathcal{V}} %variable node sets
\newcommand{\calJ }{\mathcal{J}} %constraint node sets
\newcommand{\calW }{\mathcal{W}}
\newcommand{\calB }{\mathcal{B}}
\newcommand{\calP }{\mathcal{P}}

\newcommand{\dL }{d_\mathrm{L}}
\newcommand{\dR }{d_\mathrm{R}}

\renewcommand{\Pr }{\mathrm{Pr}}

\newcommand{\PPrefix }{\mathrm{Prefix}^+}
\newcommand{\aug }{\mathrm{aug}}

\newcommand{\config }{\mathrm{vconfig}}
\newcommand{\sub }{\mathrm{sub}}

\newcommand{\LP}{\textsc{lp}}
\newcommand{\ML}{\textsc{ml}}

\newcommand{\NWMS}{\textsc{nwms}}
\newcommand{\certNWMS}{\textsc{ML-certified-nwms}}
\newcommand{\VERIFY}{\textsc{verify-lo}}
\newcommand{\DEC}{\textsc{dec}}

\newcommand{\calCj }{\mathcal{C}^j}
\newcommand{\calCbar }{\overline{\mathcal{C}}}
\newcommand{\calCbarJ}{\overline{\mathcal{C}}^{\mathcal{J}}}
\newcommand{\calCJ }{\mathcal{C}^\mathcal{J}}

\newcommand{\pnw}{projected normalized weighted}
\newcommand{\PNW}{PNW}

\newcommand{\ignore}[1]{}

\newenvironment{remark}[0]{\noindent \textbf{   Remark:}~}{\endproof}

\newtheorem{definition}{Definition}
\newtheorem{lemma}[definition]{Lemma}
\newtheorem{claim}[definition]{Claim}
\newtheorem{proposition}[definition]{Proposition}
\newtheorem{theorem}[definition]{Theorem}
\newtheorem{corollary}[definition]{Corollary}

\begin{document}

\title{On Decoding Irregular Tanner Codes with Local-Optimality Guarantees\thanks{The material in this paper was presented in part at the 2012 IEEE International Symposium on Information Theory, Cambridge, MA, USA, Jul. 2012, and in part at the 2012 IEEE 27th Convention of Electrical and Electronics Engineers in Israel, Eilat, Israel, Nov. 2012.}}

\author{
      Nissim Halabi \thanks{School of Electrical Engineering, Tel-Aviv University, Tel-Aviv 69978, Israel. \mbox{{E-mail}:\ {\tt nissimh@eng.tau.ac.il}.}}
      \and
      Guy Even \thanks{School of Electrical Engineering, Tel-Aviv University, Tel-Aviv 69978, Israel.  \mbox{{E-mail}:\ {\tt guy@eng.tau.ac.il}.}}}

\date{}

 \maketitle

\begin{abstract}
  We consider decoding of binary linear Tanner codes using message-passing
  iterative decoding and linear programming (LP) decoding in
  memoryless binary-input output symmetric (MBIOS) channels.
  We present new certificates that are based on a combinatorial characterization for local-optimality of a codeword in irregular Tanner codes with respect to any MBIOS
  channel.
  This characterization is a generalization of [Arora, Daskalakis, Steurer, \emph{Proc. ACM Symp. Theory of Computing}, 2009]  and [Vontobel, \emph{Proc. Inf. Theory and Appl. Workshop}, 2010] and is based on a conical
  combination of normalized weighted subtrees in the computation trees of the Tanner graph.
  These subtrees may have any finite height $h$ (even equal or greater than half of the girth of the Tanner graph). In addition, the degrees of local-code nodes in these subtrees are not restricted to two (i.e., these subtrees are not restricted to skinny trees).
  We prove that local optimality in this new characterization implies maximum-likelihood (ML) optimality and LP optimality, and show that a certificate can be computed efficiently.

  We also present a new message-passing iterative decoding algorithm, called \emph{normalized weighted min-sum} (\NWMS). \NWMS\ decoding is a belief-propagation (BP) type algorithm that applies to any irregular binary Tanner code with single parity-check local codes (e.g., LDPC codes and HDPC codes).
  We prove that if a locally-optimal codeword with respect to height parameter $h$ exists (whereby notably $h$ is not limited by the girth of the Tanner graph), then \NWMS\ decoding finds this codeword in $h$ iterations.
  The decoding guarantee of the \NWMS\ decoding algorithm applies whenever there exists a locally optimal codeword.
  Because local optimality of a codeword implies that it is the unique ML codeword, the decoding guarantee also provides an ML certificate for this codeword.

  Finally, we apply the new local optimality characterization to regular Tanner codes, and prove lower bounds on the noise thresholds of LP decoding in MBIOS channels. When the noise is below these lower bounds, the probability that LP decoding fails to decode the transmitted codeword decays doubly exponentially in the girth of the Tanner graph.
\end{abstract}

\clearpage
%%%%%%%%%%%%%%%%%%%%%%%%%%%%%%%%%%%%
\section{Introduction} \label{sec:intro}
%%%%%%%%%%%%%%%%%%%%%%%%%%%%%%%%%%%%

Modern coding theory deals with finding good error correcting codes that have efficient
encoders and decoders~\cite{RU08}. Message-passing iterative decoding algorithms based on
belief propagation (see, e.g., \cite{Gal63,BGT93,Mac99,LMSS01,RU01}) and linear programming (LP) decoding~\cite{Fel03,FWK05} are examples of efficient decoders. These decoders are usually suboptimal, i.e., they may fail to correct errors that are corrected by a maximum likelihood (ML) decoder.

Many works deal with low-density parity-check (LDPC) codes and generalizations of LDPC codes. LDPC codes were first defined by Gallager~\cite{Gal63} who suggested several message-passing iterative decoding algorithms (including an algorithm that is now known as the sum-product decoding algorithm). Tanner~\cite{Tan81} introduced graph representations (nowadays known as Tanner graphs) of linear codes based on bipartite graphs over variable nodes and constraint nodes, and viewed iterative decoding as message-passing algorithms over the edges of these bipartite graphs. In the standard setting, constraint nodes enforce a zero-parity among their neighbors. In the generalized setting, constraint nodes enforce a local error-correcting
code. One may view a constraint node with a linear local code as a coalescing of
multiple single parity-check nodes. Therefore, a code may have a sparser and
smaller representation when represented as a Tanner code in the generalized setting.
Sipser and Spielman~\cite{SS96} studied binary Tanner codes based on expander graphs and analyzed a simple bit-flipping decoding algorithm.

Wiberg \emph{et al.}~\cite{WLK95,Wib96} developed the use of graphical models for systematically describing instances of known decoding algorithms. In particular, the sum-product decoding algorithm and the min-sum decoding algorithm are presented as generic iterative message-passing decoding algorithms that apply to any graph realization of a Tanner code. Wiberg \emph{et al.}\ proved that the min-sum decoding algorithm can be viewed as a dynamic programming algorithm that computes the ML codeword if the Tanner graph is a tree.
For LDPC codes, Wiberg~\cite{Wib96} characterized a necessary condition for decoding failures of the min-sum decoding algorithm by ``negative'' cost trees, called \emph{minimal deviations}.

LP decoding was introduced by Feldman, Wainwright, and Karger~\cite{Fel03,FWK05} for binary linear codes. LP decoding is based on solving a fractional relaxation of an integer linear program that models the problem of ML decoding. The vertices of the relaxed LP polytope are called \emph{pseudocodewords}. Every codeword is a vertex of the relaxed LP polytope, however, usually there are additional vertices for which at least one component is non-integral.
LP decoding has been applied to several codes, among them: cycle codes, turbo-like and RA codes~\cite{FK04,HE05,GB11}, LDPC codes~\cite{FMSSW07,DDKW08,KV06,ADS09,HE11}, and expander codes~\cite{FS05,Ska11}. Our work is motivated by the problem of finite-length and average-case analysis of successful LP decoding of binary Tanner codes. There are very few works on this problem, and they deal only with specific cases. For example, Feldman and Stein~\cite{FS05} analyzed
special expander-based codes, and Goldenberg and Burshtein~\cite{GB11} dealt with repeat-accumulate codes.

\paragraph{Previous results.}
Combinatorial characterizations of sufficient conditions for successful decoding of the ML codeword are based on so called ``certificates.'' That is,
given a channel observation $y$ and a codeword $x$, we are interested in a one-sided error test that
answers the questions: is $x$ optimal with respect to $y$? is it
unique? Note that the test may answer ``no'' for a positive instance. A positive answer for such a test is called a \emph{certificate} for the optimality of a codeword. Upper bounds on the word error probability are obtained by lower bounds on the probability that a certificate exists.

Koetter and Vontobel~\cite{KV06} analyzed LP decoding of regular LDPC
codes. Their analysis is based on decomposing each codeword (and
pseudocodeword) to a finite set of minimal structured trees (i.e.,
skinny trees) with uniform vertex weights.  Arora \emph{et
  al.}~\cite{ADS09} extended the work in~\cite{KV06} by introducing
non-uniform weights to the vertices in the skinny trees, and defining
\emph{local optimality}.  For a BSC, Arora \emph{et al.}\ proved that
local optimality implies both ML optimality and LP optimality.  They
presented an analysis technique that performs finite-length density
evolution of a min-sum process to prove bounds on the probability of a
decoding error. Arora \emph{et al.}\ also pointed out that it is
possible to design a re-weighted version of the min-sum decoder for
regular codes that finds the locally-optimal codeword if such a codeword
exists for trees whose height is at most half of the girth of the Tanner graph.  This
work was further extended in~\cite{HE11} to memoryless binary-input output-symmetric (MBIOS) channels beyond the BSC. The
analyses presented in these works~\cite{KV06,ADS09,HE11} are limited
to skinny trees, the height of which is bounded by a half of the girth
of the Tanner graph.

Vontobel~\cite{Von10} extended the decomposition of a codeword (and a
pseudocodeword) to skinny trees in graph covers.
This enabled Vontobel to mitigate the limitation on the height of the skinny trees by half of the girth of the base Tanner graph. The
decomposition is obtained by a random walk, and applies also to irregular Tanner graphs.

Various iterative message-passing decoding algorithms have been
derived from the belief propagation algorithm (e.g., max-product
decoding algorithm~\cite{WLK95}, attenuated max-product~\cite{FK00},
tree-reweighted belief-propagation~\cite{WJW05}, etc.).  The
convergence of these belief-propagation (BP) based iterative decoding algorithms to an
optimum solution has been studied extensively in various settings
(see, e.g., \cite{WLK95,FK00,WF01,CF02,CDEFH05,WJW05,RU01,JP11}).
However, bounds on the running time required to decode (or on the
number of messages that are sent) have not been proven for these
algorithms. The analyses of convergence in these works often rely on the existence of
a single optimal solution in addition to other conditions such as:
single-loop graphs, large girth, large reweighting coefficients,
consistency conditions, etc.

Jian and Pfister~\cite{JP11} analyzed a special case of the attenuated max-product decoder~\cite{FK00} for regular LDPC codes. They considered skinny trees in the computation tree, the height of which is equal or greater than half of the girth of the
Tanner graph. Using contraction properties and consistency conditions, they proved sufficient conditions under which the message-passing decoder converges to a locally optimal codeword. This convergence also implies convergence to the LP optimum and therefore to the ML codeword.

While local-optimality characterizations were investigated for the
case of finite-length analysis of regular LDPC
codes~\cite{ADS09,HE11,JP11}, no local-optimality characterizations have been stated
for the general case of Tanner codes.
In this paper we study a generalization of previous local-optimality characterizations, and the guarantees it provides for successful ML decoding by LP decoding and iterative message-passing decoding algorithms.
%\begin{comment}
In particular, this paper presents a
decoding algorithm for finite-length (regular and irregular) LDPC
codes over MBIOS channels
with bounded time complexity that combines two properties:
\begin{inparaenum}[(i)]
%\begin{enumerate}[(i)]
\item it is a message-passing algorithm, and
\item for every number of iterations (not limited by any function of the girth of the Tanner
  graph), if the local-optimality characterization is satisfied for some codeword, then the algorithm succeeds to decode the ML codeword and has an ML certificate.
%\end{enumerate}
\end{inparaenum}
%\end{comment}

\paragraph{Contributions.}
We present a new combinatorial characterization for local optimality of a codeword in irregular binary Tanner codes with respect to (w.r.t.) any MBIOS channel (Definition~\ref{def:localOptimality}).
Local optimality is characterized via costs of deviations based on subtrees in computation trees of the Tanner graph\footnote{We consider computation trees that correspond to a ``flooding message update schedule'' in the context of iterative message-passing algorithms such as the max-product decoding algorithm.}. Consider a computation tree with height $2h$ rooted at some variable node. A deviation is based on a subtree such that (i)~the degree of a variable node is equal to its degree in the computation tree, and (ii)~the degree of a local-code node equals some constant $d\geqslant2$, provided that $d$ is at most the minimum distance of the local codes. Furthermore, level weights $w\in\R_+^h$ are assigned to the levels of the tree. Hence, a deviation is a combinatorial structure that has three main parameters: deviation height $h$, deviation level weights $w\in\R_+^h$, and deviation ``degree'' $d$. Therefore, the new definition of local optimality is based on three parameters: $h\in\N$, $w\in\R_+^h$, and $d\geqslant2$.

This characterization extends the notion of deviations in local optimality in four ways:
\begin{compactenum}[(i)]
%\begin{enumerate}[(i)]
\item no restrictions are applied to the degrees of the nodes in the Tanner graph,
\item arbitrary local linear codes may be associated with constraint nodes,
\item deviations are subtrees in the computation tree and no limitation is set on the height of the deviations; in particular, their height may exceed (any function of) the girth of the Tanner graph, and
\item deviations may have a degree $d\geqslant 2$ in the local-code nodes (as opposed to skinny trees in previous analyses), provided
  that $d$ is at most the minimum distance of the local codes.
%\end{enumerate}
\end{compactenum}
We prove that local optimality in this new characterization implies
ML optimality (Theorem~\ref{thm:MLsufficient}). We utilize the
equivalence of graph cover decoding and LP decoding for Tanner codes,
implied by Vontobel and Koetter~\cite{VK05} to prove that
local optimality suffices also for LP optimality
(Theorem~\ref{thm:LPsufficient}). We present an efficient dynamic
programming algorithm that computes a local-optimality certificate, and hence an ML certificate\footnote{An ML certificate computed based on local optimality is different from an ML certificate computed by LP decoding~\cite{FWK05} in the following sense.
In the context of LP decoding, the ML certificate property means that if the LP decoder outputs an integral word, then it must be the ML codeword. Hence, one may compute an ML certificate for a codeword $x$ and a given channel output by running the LP decoder and compare its result with the codeword $x$. Local optimality is a combinatorial characterization of a codeword with respect to an LLR vector, which, by Theorem~\ref{thm:MLsufficient}, suffices for ML. Hence, one may compute an ML certificate for a codeword $x$ and a given channel output by verifying that the codeword is locally optimal w.r.t. the channel output. Algorithm~\ref{alg:verify} is an efficient message-passing algorithm that returns true if the codeword is locally optimal, and therefore provides an ML certificate.}, for
a codeword w.r.t. a given channel output (Algorithm~\ref{alg:verify}), if such certificate exists.

We present a new message-passing iterative decoding algorithm, called
\emph{normalized weighted min-sum} (\NWMS) decoding algorithm
(Algorithm~\ref{alg:weighted-min-sum}). The \NWMS\ decoding algorithm applies to
any irregular Tanner code with single parity-check (SPC) local codes (e.g.,
LDPC codes and HDPC codes).
The input to the \NWMS\ decoding algorithm consists of the channel output and two additional parameters that characterize local optimality for Tanner codes with SPC local codes: (i)~a certificate height $h$, and (ii)~a vector of layer weights $w \in \R_+^h\setminus\{0^h\}$. (Note that the local codes are SPC codes, and therefore the deviation degree $d$ equals $2$.)
We prove that, for any finite $h$, the \NWMS\ decoding algorithm is
guaranteed to compute the ML codeword in $h$ iterations if an
$h$-locally-optimal codeword exists (Theorem~\ref{thm:MPsufficient}).
The decoding guarantee of the \NWMS\ algorithm is not bounded by (any function of) the girth.
Namely, the height parameter $h$ in local optimality and the number of
iterations in the decoding is arbitrary and may exceed (any function of) the girth.
Because local optimality is a pure combinatorial property, the decoding guarantee of the \NWMS\ decoding algorithm is not asymptotic nor does it rely on convergence. Namely, it applies to finite codes and decoding with a finite number of iterations.
Furthermore, the output of the \NWMS\ decoding algorithm can be ML-certified efficiently (Algorithm~\ref{alg:verify}).
The time and message complexity of the \NWMS\ decoding algorithm is
$O(|E|\cdot h)$ where $|E|$ is the number of edges in the Tanner
graph. Local optimality, as defined in this paper, is a sufficient
condition for successfully decoding the unique ML codeword by our BP-based algorithm in loopy
graphs.

Previous bounds on the probability that a local-optimality certificate
exists~\cite{KV06,ADS09,HE11} hold for regular LDPC codes. The same bounds
hold also for successful decoding of the transmitted codeword by the \NWMS\ decoding algorithm.
These bounds are based on proving that a local-optimality certificate exists with high probability for the transmitted codeword when the noise in the channel is below some noise threshold. The resulting threshold values happen to be relatively close to the BP thresholds. Specifically, noise
thresholds of $p^*\geqslant0.05$ in the case of a BSC~\cite{ADS09}, and
$\sigma^*\geqslant0.735$ ($\frac{E_b}{N_0}^*\leqslant2.67$dB) in the case of a BI-AWGN channel~\cite{HE11} are proven for $(3,6)$-regular LDPC codes whose Tanner graphs have logarithmic girth in the block-length.

Finally, for a fixed height, trees in our new characterization contain more
vertices than a skinny tree because the internal degrees are bigger.
Hence, over an MBIOS channel, the probability of a locally-optimal certificate with dense
deviations (local-code node degrees bigger than two) is greater than
the probability of a locally-optimal certificate based on skinny trees
(i.e., local-code nodes have degree two).
This characterization leads to improved bounds for successful decoding of the transmitted codeword of regular Tanner codes (Theorems~\ref{thm:main-bound-BSC} and~\ref{thm:main-bound-MBIOS}).

We extend the probabilistic analysis of the min-sum process by Arora \emph{et al.}~\cite{ADS09} to a sum-min-sum process on regular
trees. For regular Tanner codes, we prove bounds on the word error probability of LP decoding
under MBIOS channels. These bounds are inverse doubly-exponential in the girth of the Tanner graph.
We also prove bounds on the threshold of regular Tanner codes whose Tanner graphs have
logarithmic girth. This means that if the noise in the channel is below that threshold, then
the decoding error diminishes exponentially as a function of the block length.
Note that Tanner graphs with logarithmic girth can be constructed explicitly (see, e.g.,~\cite{Gal63}).

To summarize, our contribution is threefold. %\begin{inparaenum}[(i)]
\begin{compactenum}[(i)]
\item We present a new combinatorial characterization of local optimality for binary Tanner codes w.r.t. any MBIOS channel. This characterization provides an ML certificate and an LP certificate for a given codeword. The certificate can be efficiently computed by a dynamic programming algorithm. Based on this new characterization, we present two applications of local optimality.
\item A new efficient message-passing decoding algorithm, called \emph{normalized weighted min-sum} (\NWMS), for irregular binary Tanner codes with SPC local codes (e.g., LDPC codes and HDPC codes).
The \NWMS\ decoding algorithm is guaranteed to find the locally optimal codeword in $h$ iterations, where $h$ determines the height of the local-optimality certificate. Note that $h$ is not bounded and may be larger than (any function of) the girth of the Tanner graph (i.e., decoding with local-optimality guarantee ``beyond the girth'').
\item New bounds on the word error probability are proved for LP decoding of regular binary Tanner codes.
\end{compactenum}
%\end{inparaenum}

\paragraph{Organization.}
The remainder of this paper is organized as follows. Section~\ref{sec:prelim} provides background on ML decoding and LP decoding of binary Tanner codes over MBIOS channels. Section~\ref{sec:MLcertificate} presents a combinatorial certificate that applies to ML decoding for codewords of Tanner codes. In Section~\ref{sec:LPcertificate}, we prove that the certificate applies also to LP decoding for codewords of Tanner codes. In Section~\ref{sec:verify}, we present an efficient certification algorithm for local optimality.
Section~\ref{sec:NWMS} presents the \NWMS\ iterative decoding algorithm for irregular Tanner codes with SPC local codes, followed by a proof that the \NWMS\ decoding algorithm finds the locally-optimal codeword. In Section~\ref{sec:bounds}, we use the combinatorial characterization of local optimality to bound the error probability of LP decoding for regular Tanner codes. Finally, conclusions and a discussion are given in Section~\ref{sec:conclusion}.

%%%%%%%%%%%%%%%%%%%%%%%%%%%%%%%%%%%%
\section{Preliminaries} \label{sec:prelim}
%%%%%%%%%%%%%%%%%%%%%%%%%%%%%%%%%%%%

%\paragraph{Graph Terminology.}
\subsection{Graph Terminology}
Let $G=(V,E)$ denote an undirected graph.  Let $\mathcal{N}_G(v)$ denote the set
of neighbors of node $v \in V$, and for a set $S \subseteq V$ let
$\mathcal{N}_G(S)\triangleq\bigcup_{v \in S}\mathcal{N}_G(v)$. Let $\deg_G(v)\triangleq\lvert\mathcal{N}_G(v)\rvert$ denote the edge degree of node $v$ in graph $G$.
A path
$p=(v,\ldots,u)$ in $G$ is a sequence of vertices such that there
exists an edge between every two consecutive nodes in the sequence
$p$. A path $p$ is \emph{backtrackless} if every two consecutive edges along $p$ are distinct.
Let $s(p)$ denote the first vertex (\emph{source}) of path $p$, and let $t(p)$ denote the last vertex (\emph{target}) of path $p$. If $s(p)=t(p)$ then the path is \emph{closed}. A \emph{simple} path is a path with no repeated vertex. A \emph{simple cycle} is a closed backtrackless path where the only repeated vertex is the first and last vertex. Let $\lvert p\rvert$ denote the length of a path $p$, i.e., the number of edges in $p$.
Let $d_G(r,v)$ denote the distance
(i.e., length of a shortest path) between nodes $r$ and $v$ in $G$,
and let $\girth(G)$ denote the length of the shortest cycle in $G$. Let $p$ and $q$ denote two paths in a graph $G$ such that $t(p)=s(q)$. The path obtained by \emph{concatenating} the paths $p$ and $q$ is denoted by $p\circ q$.

An \emph{induced subgraph} is a subgraph obtained by deleting a set of
vertices. In particular, the \emph{subgraph of $G$ induced by $S \subseteq V$} consists of $S$ and all edges in $E$, both endpoints of which
are contained in $S$. Let $G_S$ denote the subgraph of $G$ induced by $S$.

%\paragraph{Tanner codes and Tanner graph representation.}
\subsection{Tanner Codes and Tanner Graph Representation}
Let $G=(\mathcal{V} \cup \mathcal{J}, E)$ denote an edge-labeled
bipartite graph, where $\mathcal{V} = \{v_1,\ldots,v_N\}$ is a set of
$N$ vertices called \emph{variable nodes}, and $\mathcal{J} =
\{C_1,\ldots,C_J\}$ is a set of $J$ vertices called \emph{local-code
  nodes}. We denote the degree of $C_j$ by $n_j$.

Let $\calCbarJ \triangleq \big\{\calCbar^j~\big\vert~\calCbar^j~\mathrm{is~an}~[n_j,k_j,d_j]~ \mathrm{code},~1\leqslant j\leqslant J\big\}$ denote a
set of $J$ linear \emph{local codes}. The local code $\calCbar^j$ corresponds
to the vertex $C_j \in \calJ$.  We say that $v_i$
\emph{participates} in $\calCbar^j$ if $(v_i,C_j)$ is an edge in $E$.
The edges incident to each local-code node $C_j$ are labeled
$\{1,\ldots,n_j\}$. This labeling specifies the index of a variable
node in the corresponding local code.

A word $x = (x_1,\ldots,x_N) \in \{0,1\}^N$ is an assignment
to variable nodes in $\mathcal{V}$ where $x_i$ is assigned to $v_i$. Let $\mathcal{V}_j$ denote the set $\calN_G(C_j)$ ordered according to labels of
edges incident to $C_j$. Denote by $x_{\mathcal{V}_j} \in
\{0,1\}^{n_j}$ the projection of the word $x = (x_1,\ldots,x_N)$ onto
entries associated with $\mathcal{V}_j$.

The binary \emph{Tanner code} $\calC(G,\calCbarJ)$ based on the labeled
\emph{Tanner graph} $G$ is the set of vectors $x \in \{0,1\}^N$
such that $x_{\mathcal{V}_j}$ is a codeword in $\calCbar^j$ for every
$j \in \{1,\ldots,J\}$. Let us note that all the codes that we consider in this paper are binary linear codes.

Let $d_j$ denote the minimum distance of the local code $\calCbar^j$.
The \emph{minimum local distance} $d^*$ of a Tanner code
$\calC(G,\calCbarJ)$ is defined by $d^*=\min_j d_j$. We assume that $d^*\geq 2$.

If the bipartite graph is $(\dL,\dR)$-regular, i.e., the vertices in
$\calV$ have degree $\dL$ and the vertices in $\calJ$ have degree
$\dR$, then the resulting code is called a \emph{$(\dL,\dR)$-regular Tanner
  code}.

If the Tanner graph is sparse, i.e., $|E|=O(N)$, then it defines a
\emph{low-density Tanner code}. A \emph{single parity-check code} is a code that
contains all binary words with even Hamming weight. Tanner codes that have single parity-check local codes and that are based on sparse Tanner graphs are called
\emph{low-density parity-check (LDPC) codes}.

Consider a Tanner code $\calC(G,\calCbarJ)$. We say that a word $x =
(x_1,...,x_N)$ \emph{satisfies} the local code $\calCbar^j$ if its
projection $x_{\mathcal{V}_j}$ is in $\calCbar^j$.  The set of words
$x$ that \emph{satisfy} the local code $\calCbar^j$ is denoted by
$\calCj$, i.e., $\calCj = \{x \in \{0,1\}^N~\mid~x_{\mathcal{V}_j} \in
\calCbar^j \}$. Namely, the resulting code $\calCj$ is the
\emph{extension} of the local code $\calCbar^j$ from length $n_j$ to
length $N$. The Tanner code is simply the intersection of the
extensions of the local codes, i.e.,
\begin{equation}
\calC(G,\calCbarJ) = \bigcap_{j \in \{1,\ldots,J\}}{\calCj}.
\end{equation}

%From this point, let $d^* \triangleq \min_{j\in\mathcal{J}}d_{\min}(\calCbar^j)$.

%\paragraph{LP decoding of Tanner codes over memoryless channels.}
\subsection{LP Decoding of Tanner Codes over Memoryless Channels}
Let $c_i \in \{0,1\}$ denote the $i$th transmitted binary symbol (channel input), and let $y_i\in\R$ denote the $i$th received symbol (channel output).
A \emph{memoryless binary-input output-symmetric} (MBIOS) channel is defined by a conditional probability density function $f(y_i|c_i=a)$ for $a\in\{0,1\}$ that satisfies $f(y_i|0) = f(-y_i|1)$. The binary erasure channel (BEC), binary symmetric channel (BSC) and binary-input additive white Gaussian noise (BI-AWGN) channel are examples for MBIOS channels.
Let $y\in\R^N$ denote the word received from the channel.
In MBIOS channels, the \emph{log-likelihood ratio} (LLR) vector $\lambda=\lambda(y) \in \R^N$ is defined by $\lambda_i (y_i) \triangleq
\ln\big(\frac{f(y_i|c_i=0)}{f(y_i|c_i=1)}\big)$ for every
input bit $i$. For a code $\calC$, \emph{Maximum Likelihood (ML)
decoding} is equivalent\footnote{Strictly speaking, the operator $\argmin$ returns a set of vectors because $\langle
\lambda(y) , x \rangle$ may have multiple minima w.r.t. $\mathrm{conv}(\calC)$. When $\argmin$ returns a singleton set, then $\argmin$ is equal to the vector in that set. Otherwise, it returns a random vector from the set.} to
\begin{equation} \label{eqn:MLdecoding}
 \hat{x}^{\ML}(y) = \argmin_{x \in \mathrm{conv}(\mathcal{C})} \langle
\lambda(y) , x \rangle,
\end{equation}
where $\mathrm{conv}(\mathcal{C})$ denotes the convex hull of the set $\mathcal{C}$ where $\{0,1\}^N$ is considered to be a subset of $\R^N$.

In general, solving the optimization problem in (\ref{eqn:MLdecoding})
for linear codes is intractable~\cite{BMT78}. Feldman \emph{et
  al.}~\cite{Fel03,FWK05} introduced a linear programming relaxation
for the problem of ML decoding of Tanner codes with single parity-check codes acting as local codes. The resulting relaxation of $\mathrm{conv}(\calC)$ is nowadays called the \emph{fundamental polytope}~\cite{VK05} of the Tanner graph $G$.
We consider an extension of this definition to the case in which
the local codes are arbitrary as follows. The \emph{generalized
  fundamental polytope} $\calP \triangleq \calP(G,\calCbarJ)$ of a
Tanner code $\calC = \calC(G,\calCbarJ)$ is defined by
\begin{equation}
\calP \triangleq \bigcap_{\calCj \in \calCJ}{\mathrm{conv}(\calCj)}.
\end{equation}
Note that a code may have multiple representations by a Tanner
graph and local codes.  Moreover, different representations
$(G,\calCbarJ)$ of the same code $\calC$ may yield different
generalized fundamental polytopes $\calP(G,\calCbarJ)$. If the degree
of each local-code node is constant, then the generalized fundamental
polytope can be represented by $O(N+ |\calJ|)$ variables and $O(|\calJ|)$
constraints. Typically,
$|\calJ|=O(N)$, and the generalized fundamental polytope has an
efficient representation.  Such Tanner codes are often called
\emph{generalized low-density parity-check codes}.

Given an LLR vector $\lambda$ for a received word $y$, LP decoding is
defined by the following linear program
\begin{equation} \label{eqn:LPdecoding} \hat{x}^{\LP}(y) \triangleq
  \argmin_{x \in \mathcal{P}(G,\calCbarJ)} \langle \lambda(y) , x
  \rangle.
\end{equation}

The difference between ML decoding and LP decoding is that the
fundamental polytope $\calP(G,\calCbarJ)$ may strictly contain the
convex hull of $\calC$.
Vertices of $\calP(G,\calCbarJ)$ are called \emph{pseudocodewords}~\cite{Fel03,FWK05}. It can be shown that vertices of $\calP(G,\calCbarJ)$ that are not codewords of $\calC$ must have at least one non-integral component.

%%%%%%%%%%%%%%%%%%%%%%%%%%%%%%%%%%%%
\section{A Combinatorial Certificate for an ML Codeword} \label{sec:MLcertificate}
%%%%%%%%%%%%%%%%%%%%%%%%%%%%%%%%%%%%

In this section we present combinatorial certificates for codewords of
Tanner codes that apply both to ML decoding and LP decoding. A
certificate is a proof that a given codeword is the unique solution of
ML decoding and LP decoding. The
certificate is based on combinatorial weighted structures in the
Tanner graph, referred to as \emph{local configurations}. These local
configurations generalize the minimal configurations (skinny trees)
presented by Vontobel~\cite{Von10} as extension to Arora \emph{et
  al.}~\cite{ADS09}. We note that for Tanner codes, the characteristic function of the support of
each weighted local configuration is not necessarily a locally valid
configuration. For a given codeword, the certificate is computed by a
dynamic-programming algorithm on the Tanner graph of the code (see Section~\ref{subsec:certificateVerifiaction}).

\emph{Notation:} Let $y \in \R^N$ denote the word received from the channel.
Let $\lambda = \lambda(y)$ denote the LLR vector for $y$. Let $G=(\mathcal{V}\cup\mathcal{J},E)$ denote a Tanner graph, and let $\mathcal{C}(G)$ denote a Tanner code based on $G$ with minimum local distance $d^*$. Let $x \in
\mathcal{C}(G)$ be a candidate for $\hat{x}^{\ML}(y)$ and
$\hat{x}^{\LP}(y)$.

\begin{definition}[Path-Prefix Tree]
  Consider a graph $G=(V,E)$ and a node $r \in V$. Let $\hat{V}$
  denote the set of all backtrackless paths in $G$ with length at most
  $h$ that start at node $r$, and let
\begin{equation*}
\hat{E}\triangleq\bigg\{(p_1,p_2)\in\hat{V}\times\hat{V}~\bigg\vert~p_1~\mathrm{is~a~prefix~of}~p_2,~ \lvert p_1\rvert+1=\lvert p_2\rvert \bigg\}.
\end{equation*}
We identify the zero-length path in $\hat{V}$ with $(r)$.
Denote by $\calT_r^{h}(G) \triangleq (\hat{V},\hat{E})$ the \emph{path-prefix tree} of $G$ rooted at node $r$ with height $h$.
\end{definition}
Path-prefix trees of $G$ that are rooted at a variable node or at a local-code node are often
called \emph{computation trees}. We consider also path-prefix trees
of subgraphs of $G$ that may be either rooted at a variable node or at a
local-code node.

We use the following notation. Vertices in $G$ are denoted by $u,v,r$. Because vertices in $\calT_r^{h}(G)=(\hat{V},\hat{E})$
are paths in $G$, we denote vertices in path-prefix trees by $p$ and
$q$.  For a path
$p\in\hat{V}$, let $s(p)$ denote the first vertex (\emph{source}) of
path $p$, and let $t(p)$ denote the last vertex (\emph{target}) of path
$p$. Denote by $\PPrefix(p)$ the set of proper prefixes of the path
$p$, i.e.,
\begin{equation*}
\PPrefix(p) = \bigg\{q\ \bigg\vert\ q \mathrm{\ is\ a\ prefix\ of\ }p,\ 1 \leqslant\rvert q \lvert<\lvert p \rvert\bigg\}.
\end{equation*}
Consider a Tanner graph $G=(\calV\cup\calJ,E)$ and let
$\calT_r^h(G)=(\hat{V},\hat{E})$ denote a path-prefix tree of $G$.
Let $\hat{\calV}\eqdf \{p\ \vert\ p\in\hat{V},\ t(p)\in\calV\}$, and
$\hat{\calJ}\eqdf \{p\ \vert\ p\in\hat{V},\ t(p)\in\calJ\}$.
Paths in $\hat{\calV}$ are called
\emph{variable paths}\footnote{Vertices in a path-prefix tree of a Tanner graph $G$ correspond to paths in $G$. We therefore refer by \emph{variable paths} to vertices in a path-prefix tree that correspond to paths in $G$ that end at a variable node.}, and paths in  $\hat{\calJ}$ are called \emph{local-code paths}.

The following definitions expand the combinatorial notion of minimal valid deviations~\cite{Wib96} and weighted minimal local deviations (skinny trees)~\cite{ADS09, Von10} to the case of Tanner codes.
\begin{definition}[$d$-tree]
Consider a Tanner graph $G=(\calV\cup\calJ,E)$. Denote by $\calT_r^{2h}(G)=(\hat\calV\cup\hat\calJ,\hat{E})$ the path-prefix tree of $G$ rooted at node $r\in\calV$. A subtree $\calT \subseteq \calT_r^{2h}(G)$ is a \emph{$d$-tree} if:
%\begin{inparaenum}[(i)]
\begin{compactenum}[(i)]
\item $\calT$ is rooted at $(r)$,
\item for every local-code path $p\in\calT\cap\hat\calJ$, $\deg_\calT(p)=d$, and
\item for every variable path $p\in\calT\cap\hat\calV$, $\deg_\calT(p)=\deg_{\calT_r^{2h}}(p)$.
\end{compactenum}
%\end{inparaenum}
\end{definition}
Note that the leaves of a $d$-tree are variable paths because a $d$-tree is rooted in a variable node and has an even height.
Let $\calT[r,2h,d](G)$ denote the set of all $d$-trees rooted at $r$ that are subtrees of $\calT_r^{2h}(G)$.

In the following definition we use ``level'' weights $w=(w_1,\ldots,
w_h)$ that are assigned to variable paths in a subtree of a path-prefix tree of height $2h$.
\begin{definition} [$w$-weighted subtree]\label{def:weightedSubtree}
  Let $\calT = (\hat{\calV}\cup\hat{\calJ},\hat{E})$ denote a subtree
  of $\calT_r^{2h}(G)$, and let $w=(w_1,\ldots,w_h)\in\R_+^h$ denote a
  non-negative weight vector.  Let
  $w_\calT:\hat{\mathcal{V}}\rightarrow \R$ denote a weight function
  based on the weight vector $w$ for variable paths $p\in\hat{\calV}$
  defined as follows. If $p$ is a zero-length variable path, then $w_\calT(p)=0$. Otherwise,
\begin{equation*}%\label{eqn:w-weights}
   w_\calT(p) \triangleq \frac{w_\ell}{
\wone}\cdot\frac{1}{\deg_G\big(t(p)\big)}\cdot \prod_{q\in \PPrefix(p)}\frac{1}{\deg_{\calT}(q)-1},
\end{equation*}
where $\ell = \big\lceil\frac{\lvert p\rvert}{2}\big\rceil$. We refer to $w_\calT$ as a $w$-weighed subtree.
\end{definition}

For any $w$-weighted subtree $w_\calT$ of $\calT_{r}^{2h}(G)$, let $\pi_{G,\calT,w}:\calV \rightarrow \R$ denote a function whose values correspond to the projection of $w_\calT$ to the Tanner graph $G$. That is, for every variable node $v$ in $G$,
\begin{equation}
\pi_{G,\calT,w}(v) \triangleq \sum_{\{p\in\calT~\vert~t(p)=v\}}w_\calT(p).
\end{equation}
We remark that:
\begin{inparaenum}[(i)]
\item If no variable path in $\calT$ ends in $v$, then
  $\pi_{G,\calT,w}(v)=0$.
\item If $h<\girth(G)/4$, then every node $v$ is an endpoint of at
  most one variable path in $\calT_{r}^{2h}(G)$, and the projection is
  trivial. However, we deal with arbitrary heights $h$, in which case the projection is many-to-one since many  different variable paths may share a common endpoint.
\end{inparaenum}
Notice that the length of the weight vector $w$ equals the height parameter $h$.

\begin{definition}\label{def:PNWdevSet}
Consider a Tanner code $\calC(G)$, a non-positive weight vector $w\in\R_+^h$, and $2\leqslant d\leqslant d^*$. Let $\mathcal{B}_d^{(w)}$ denote the set of all projections of
$w$-weighted $d$-trees to $G$, i.e.,
\[\mathcal{B}_d^{(w)} \triangleq \bigg\{\frac{1}{c}\cdot\pi_{G,\calT,w}\ \bigg\vert\ \calT\in\bigcup_{r\in\calV}\calT[r,2h,d](G)\bigg\},\]
where $c\geqslant1$ is chosen so that $\mathcal{B}_d^{(w)}\subseteq
[0,1]^N$.
\end{definition}
Vectors in $\mathcal{B}_d^{(w)}$ are referred to as \emph{\pnw\ (\PNW) deviations}.
We use a \PNW\ deviations to alter a codeword in the upcoming definition of local optimality (Definition~\ref{def:localOptimality}). Our notion of deviations differs from Wiberg's deviations~\cite{Wib96} in three significant ways:
\begin{compactenum}[(i)]
%\begin{enumerate}[(i)]
\item For a $d$-tree $\calT$, the characteristic function of the support of $w_\calT$ is not necessarily a valid configuration of the computation tree.
\item The entries of $w_\calT$ are real scaled version of the characteristic function of the support of $w_\calT$. The scaling obeys a degree normalization along the path from the root of $\calT$ and a non-negative level weight factor as extension of weighted minimal deviations~\cite{ADS09,Von10}.
\item  We apply a projection operator $\pi$ on $w_\calT$ to the Tanner graph $G$. The characteristic function of the support of the projection does not induce a tree on the Tanner graph $G$ when $h$ is large.
%\end{enumerate}
\end{compactenum}

For two vectors $x \in \{0,1\}^N$ and $f \in [0,1]^N$, let $x\oplus f \in [0,1]^N$ denote the \emph{relative point} defined by $(x\oplus f)_i \triangleq |x_i-f_i|$~\cite{Fel03}. The following definition is an extension of local optimality~\cite{ADS09,Von10} to Tanner codes on memoryless channels.

\begin{definition}[local optimality] \label{def:localOptimality} Let
  $\mathcal{C}(G) \subset \{0,1\}^N$ denote a Tanner code with minimum
  local distance $d^*$. Let $w \in \R_+^h\backslash\{0^h\}$ denote a
  non-negative weight vector of length $h$ and let $2 \leqslant d
  \leqslant d^*$.  A codeword $x \in \calC(G)$ is
  \emph{$(h,w,d)$-locally optimal w.r.t. $\lambda \in \R^N$}
  if for all vectors $\beta \in \mathcal{B}_d^{(w)}$,
\begin{equation}
\langle \lambda,x \oplus \beta \rangle > \langle \lambda, x \rangle.
\end{equation}
\end{definition}

Based on random walks on the Tanner graph, the results in~\cite{Von10} imply that
$(h,w,d\!=\!2)$-local optimality is sufficient both for ML optimality and
LP optimality. The transition probabilities of these random walks are induced by pseudocodewords of the generalized fundamental polytope.  We extend the results of
Vontobel~\cite{Von10} to ``thicker'' sub-trees by using probabilistic
combinatorial arguments on graphs and the properties of graph cover
decoding~\cite{VK05}.  Specifically, for any $d$ with $2 \leqslant d \leqslant d^*$ we prove that $(h,w,d)$-local optimality for a codeword $x$ w.r.t. $\lambda$ implies both ML
and LP optimality for a codeword $x$ w.r.t. $\lambda$ (Theorems~\ref{thm:MLsufficient} and~\ref{thm:LPsufficient}).

The following structural lemma states that every codeword of a Tanner code is a finite conical combination of projections of weighted trees in the computation trees of $G$.
\begin{lemma}[conic decomposition of a codeword] \label{lemma:IntegralDecomposition}
Let $\mathcal{C}(G)$ denote a Tanner code with minimum local distance $d^*$, and let $h$ be some positive integer. Consider a codeword $x \neq 0^N$. Then, for every $2 \leqslant d \leqslant d^{*}$, there exists a distribution $\rho$ over $d$-trees of $G$ of height $2h$
  such that for every weight vector $w \in\R_+^h\setminus\{0^h\}$, it
  holds that
\[x=\lVert x\rVert_1\cdot\E_\rho\big[\pi_{G,\calT,w}\big].\]
\end{lemma}
\begin{proof}
See Appendix~\ref{sec:decomposition}.
\end{proof}

Given Lemma~\ref{lemma:IntegralDecomposition}, the following theorem is obtained by modification of the proof of~\cite[Theorem 2]{ADS09} or ~\cite[Theorem 6]{HE11}.
\begin{theorem}[local optimality is sufficient for ML]\label{thm:MLsufficient}
Let $\mathcal{C}(G)$ denote a Tanner code with minimum local distance $d^*$. Let $h$ be some positive integer and $w=(w_1,\ldots,w_h)\in\R_+^h$ denote a non-negative weight vector. Let $\lambda\in\R^N$ denote the LLR vector received from the channel. If $x$ is an $(h,w,d)$-locally optimal codeword w.r.t. $\lambda$ and some $2\leqslant d \leqslant d^{*}$, then $x$ is also the unique ML codeword w.r.t. $\lambda$.
\end{theorem}
  \begin{proof}%{Theorem \ref{thm:MLsufficient}}
    We use the decomposition implied by Lemma~\ref{lemma:IntegralDecomposition} to show that for every codeword $x' \neq x$, $\langle \lambda , x'
    \rangle > \langle \lambda , x \rangle$.
    Let $z \triangleq x
    \oplus x'$. By linearity, it holds that $z\in\calC(G)$. Moreover, $z\neq0^N$ because $x\neq x'$. Because $d^*\geqslant2$, it follows that $\lVert z\rVert_1\geqslant2$. By Lemma~\ref{lemma:IntegralDecomposition} there exists a
    distribution $\rho$ over the set $\mathcal{B}_d^{(w)}$ of \PNW\ deviations such that
    $\E_{\rho} [c\cdot\beta] = \frac{z}{\lVert z\rVert_1}$, where $c\geqslant1$ is the normalizing constant so that $\mathcal{B}_d^{(w)}\subseteq[0,1]^N$ (see Definition~\ref{def:PNWdevSet}). Let $\alpha\triangleq\frac{1}{c\cdot\lVert z\rVert_1}<1$. Let $f:[0,1]^N
    \rightarrow \R$ be the affine linear function defined by
    $f(\beta) \triangleq \langle \lambda , x \oplus \beta \rangle = \langle
    \lambda , x \rangle + \sum_{i=1}^{N}(-1)^{x_i}\lambda_i \beta_i$.
    Then,
    \begin{eqnarray*}
      \langle \lambda , x \rangle &<& \E_{\rho}\langle\lambda,x\oplus\beta \rangle~~~~~~~~~~~~~~~~~~~~~(\mathrm{by~local~optimality~of}~x)\\
      &=& \langle \lambda , x \oplus \E_{\rho} \beta \rangle~~~~~~~~~~~~~~~~~~~~~(\mathrm{by~linearity~of}~f~\mathrm{and}~\E_{\beta})\\
      &=& \langle \lambda , x \oplus \alpha z \rangle~~~~~~~~~~~~~~~~~~~~~~~~(\mathrm{by~Lemma~\ref{lemma:IntegralDecomposition}})\\
      &=& \langle \lambda , (1-\alpha)x+\alpha(x\oplus z) \rangle\\
      &=& \langle \lambda , (1-\alpha)x+\alpha x' \rangle\\
      &=& (1-\alpha)\langle \lambda , x \rangle+\alpha\langle \lambda , x' \rangle.
    \end{eqnarray*}
    which implies that $\langle \lambda , x' \rangle>\langle \lambda , x \rangle$ as desired.
  \end{proof}

%%%%%%%%%%%%%%%%%%%%%%%%%%%%%%%%%%%%%%%%%%%%%%%%%%%%%%%%%%%%%%%%%%%%%%%%
\section{Local Optimality Implies LP Optimality}\label{sec:LPcertificate}
%%%%%%%%%%%%%%%%%%%%%%%%%%%%%%%%%%%%%%%%%%%%%%%%%%%%%%%%%%%%%%%%%%%%%%%%

In order to prove a sufficient condition for LP optimality, we
consider graph cover decoding introduced by Vontobel and
Koetter~\cite{VK05}. We note that the characterization of graph cover
decoding and its connection to LP decoding can be extended to the case
of Tanner codes in the generalized setting (see, e.g.,~\cite[Theorem 25]{Von10a} and~\cite[Theorem 2.14]{Hal12}).

We use the terms and notation of Vontobel and Koetter \cite{VK05} (see also~\cite[Appendix A]{HE11}) in
the statements of Proposition~\ref{prop:coverOptimality} and Lemma~\ref{lemma:cover Opt}. Specifically, let $\tilde{G}$ denote an $M$-cover of $G$.
Let $\tilde{x}=x^{\uparrow M}\in \mathcal{C}(\tilde{G})$  and $\tilde{\lambda}=\lambda^{\uparrow
  M} \in \R^{N\cdot M}$ denote the $M$-lifts of $x$ and $\lambda$,
respectively.

In this section we consider the following setting.  Let
$\mathcal{C}(G)$ denote a Tanner code with minimum local distance
$d^*$. Let $w \in \R_+^h\backslash\{0^h\}$ for some positive integer
$h$ and let $2\leqslant d \leqslant d^{*}$.

\begin{proposition}[local optimality of all-zero codeword is preserved by
  $M$-lifts] \label{prop:coverOptimality}
$0^N$ is an $(h,w,d)$-locally optimal codeword w.r.t.
  $\lambda \in \R^N$ if and only if $0^{N\cdot M}$ is an $(h,w,d)$-locally optimal
  codeword w.r.t. $\tilde{\lambda}$.
\end{proposition}
\begin{proof}
  Consider the surjection $\varphi$ of $d$-trees in the path-prefix
  tree of $\tilde G$ to $d$-trees in the path-prefix tree of $G$. This
  surjection is based on the covering map between $\tilde G$ and $G$. Given a \PNW\ deviation
  $\tilde\beta\eqdf \pi_{\tilde{G},\calT, w}$ based on a $d$-tree $\calT$ in the path-prefix
  tree of $\tilde G$, let $\beta\eqdf
  \pi_{G,\varphi(\calT),w}$. The proposition follows because $\langle \lambda, \beta
  \rangle =\langle \tilde\lambda, \tilde\beta \rangle$.
\end{proof}
For two vectors $y,z\in \R^N$, let ``$\ast$'' denote coordinatewise
multiplication, i.e., $y\ast z \triangleq (y_1\cdot z_1,\ldots,
y_N\cdot z_N)$. For a word $x\in\{0,1\}^N$, let $(-1)^x\in\{\pm1\}^N$
denote the vector whose $i$th component equals $(-1)^{x_i}$.
\begin{lemma}\label{lemma:isoLO}
For every $\lambda\in\R^N$ and every $\beta\in[0,1]^N$,
\begin{equation}
\langle (-1)^x\ast\lambda,\beta\rangle  = \langle  \lambda,x\oplus\beta\rangle-\langle\lambda,x\rangle.
\end{equation}
\end{lemma}

\begin{proof}
  For $\beta\in[0,1]^N$, it holds that $\langle \lambda , x \oplus
  \beta \rangle = \langle \lambda , x \rangle +
  \sum_{i=1}^{N}(-1)^{x_i}\lambda_i \beta_i$. Hence,
\begin{align*}
\langle \lambda , x \oplus \beta \rangle - \langle \lambda , x \rangle &= \sum_{i=1}^{N}(-1)^{x_i}\lambda_i \beta_i\\
&=  \langle (-1)^x\ast\lambda,\beta\rangle.
\end{align*}
\end{proof}

The following proposition states that the mapping
$(x,\lambda)\mapsto(0^N,(-1)^x\ast\lambda)$ preserves local optimality.

\begin{proposition}[symmetry of local optimality]\label{proposition:LOsymmetry}
  For every $x\in\calC$, $x$ is $(h,w,d)$-locally optimal w.r.t.
  $\lambda$ if and only if $0^N$ is $(h,w,d)$-locally optimal w.r.t.
  $(-1)^x\ast\lambda$.
\end{proposition}
\begin{proof}
By Lemma~\ref{lemma:isoLO},
$\langle \lambda,x \oplus \beta \rangle - \langle \lambda, x \rangle  = \langle (-1)^x\ast\lambda,\beta\rangle -\langle (-1)^x\ast\lambda,0^N\rangle$.
\end{proof}

\medskip \noindent
The following lemma states that
local optimality is preserved by lifting to an
$M$-cover.
\begin{lemma}\label{lemma:cover Opt}
  $x$ is $(h,w,d)$-locally optimal w.r.t. $\lambda$ if and only if
  $\tilde x$ is $(h,w,d)$-locally optimal w.r.t. $\tilde\lambda$.
\end{lemma}

\begin{proof}
  Assume that $\tilde{x}$ is a $(h,w,d)$-locally optimal codeword w.r.t.
  $\tilde{\lambda}$.  By Proposition~\ref{proposition:LOsymmetry},
  $0^{N\cdot M}$ is $(h,w,d)$-locally optimal w.r.t.
  $(-1)^{\tilde{x}} \ast \tilde{\lambda}$.  By
  Proposition~\ref{prop:coverOptimality}, $0^N$ is $(h,w,d)$-locally optimal
  w.r.t. $\big((-1)^x\ast \lambda\big)$. By Proposition~\ref{proposition:LOsymmetry}, $x$ is
  $(h,w,d)$-locally optimal w.r.t. $\lambda$. Each of these implications is
  necessary and sufficient, and the lemma follows.\end{proof}

The following theorem is obtained as a corollary of
Theorem~\ref{thm:MLsufficient} and Lemma~\ref{lemma:cover Opt}.
The proof is based on a reduction stating that if local optimality is sufficient for ML optimality, then it also suffices for LP optimality. The reduction is based on the equivalence of LP decoding and graph-cover decoding~\cite{VK05}, and follows the line of the proof of~\cite[Theorem 8]{HE11}.
\begin{theorem}[local optimality is sufficient for LP optimality]\label{thm:LPsufficient}
  If $x$ is an $(h,w,d)$-locally optimal codeword w.r.t. $\lambda$, then
  $x$ is also the unique optimal LP solution given $\lambda$.
\end{theorem}

%%%%%%%%%%%%%%%%%%%%%%%%%%%%%%%%%%%%%%%%%%%%%%%%%%%%%%%%%%%%%%%%%%%%%%%%
\section{Verifying Local Optimality}
\label{subsec:certificateVerifiaction}\label{sec:verify}
%%%%%%%%%%%%%%%%%%%%%%%%%%%%%%%%%%%%%%%%%%%%%%%%%%%%%%%%%%%%%%%%%%%%%%%%

In this section we address the problem of how to verify whether a
codeword $x$ is $(h,w,d)$-locally optimal w.r.t. $\lambda$.
By Proposition~\ref{proposition:LOsymmetry}, this is equivalent to verifying
whether $0^N$ is $(h,w,d)$-locally optimal w.r.t. $(-1)^x\ast
\lambda$, where $[(-1)^x]_i\eqdf (-1)^{x_i}$.

The verification algorithm is listed as Algorithm~\ref{alg:verify}. It
applies dynamic programming to find, for every variable node $v$, a
$d$-tree $\calT_v$, rooted at $v$, that minimizes the cost $\langle
(-1)^x\ast \lambda, \pi_{G,\calT_v,w}\rangle$. The algorithm returns
false if and only if it finds a \PNW\ deviation with non-positive
cost.  Note that the verification algorithm only computes the sign of
$\min_\beta \left( \langle \lambda,x\oplus\beta \rangle - \langle
  \lambda,x\rangle \right)$.
Moreover, the sign of $\min_\beta \left( \langle \lambda,x\oplus\beta \rangle - \langle
  \lambda,x\rangle \right)$ is invariant under scaling $\beta$ by any positive constant. Because $\lVert w\rVert_1$ contains a ``global'' information, the division of $\mu_v$ by $\lVert w\rVert_1$ does not take place to maintain the property that the verification algorithm is a
distributed message passing algorithm.

The algorithm is presented as a message passing algorithm. In every
step, a node propagates to its parent the minimum cost of the
$d$-subtree that hangs from it based on the minimum values received
from its children.  The message complexity of
Algorithm~\ref{alg:verify} is $O(|E|\cdot h)$, where $E$ denotes the
edge set of the Tanner graph. Algorithm~\ref{alg:verify} can be
implemented so that the running time of each iteration is:
\begin{inparaenum}[(i)]
\item $O(|E|)$ for the computation of the messages from variable nodes to check nodes, and
\item $O(|E|\cdot \log d)$ for the computation of the messages from check nodes to variable nodes. \end{inparaenum}

The following notation is used in Line~8 of the algorithm. For a set
$S$ of real values, let $\min^{[i]}\{S\}$ denote the $i$th smallest
member in $S$.

  \begin{algorithm}
    \begin{algorithmic}[1]
      \caption{\VERIFY$(x,\lambda,h,w,d)$ - An iterative verification
        algorithm.  Let $G=(\calV \cup \calJ, E)$ denote a Tanner
        graph.  Given an LLR vector $\lambda \in \R^{|\calV|}$, a
        codeword $x\in\calC(G)$, level weights $w \in \R_+^{h}$, and a
        degree $d\in\N_+$, outputs ``\emph{true}'' if $x$ is
        $(h,w,d)$-locally optimal w.r.t. $\lambda$; otherwise, outputs
        ``\emph{false}.''  }
      \label{alg:verify}
      \STATE Initialize:~$\forall v\in\calV$~:~$\lambda_v' \gets
      \lambda_v \cdot (-1)^{x_v}$
      \STATE~~~~~~~~~~~~~~~~~$\forall C\in\calJ$, $\forall v\in\calN(C)$:
      $\mu_{C\rightarrow v}^{(-1)} \gets 0$ \FOR {$l=0$ to $h-1$}
      \FORALL{$v\in\calV$, $C\in\calN(v)$}
      \STATE $\mu_{v\rightarrow C}^{(l)} \gets \frac{w_{h-l}}{\deg_G(v)}\lambda'_v+\frac{1}{\deg_G(v)-1}
      \sum_{C'\in\calN(v)\setminus\{C\}}\mu_{C'\rightarrow v}^{(l-1)}$
      \ENDFOR
      \FORALL{$C\in\calJ$, $v\in\calN(C)$}
      \STATE \mbox{$\mu_{C\rightarrow
        v}^{(l)}\gets\frac{1}{d-1}\cdot
      \sum_{i=1}^{d-1}\min^{[i]}\bigg\{\mu_{v'\rightarrow C}^{(l)}~\bigg\vert~ v'\in\calN(C)\setminus\{v\}\bigg\}$}
      \ENDFOR
      \ENDFOR
      \FORALL{$v\in\calV$} \STATE $\mu_v \gets
      \sum_{C\in\calN(v)}\mu_{C\rightarrow v}^{(h-1)}$ \IF[min-cost
      $w$-weighted $d$-tree rooted at $v$ has non-positive
      value]{$\mu_v \leqslant 0$}
      \RETURN \FALSE;
      \ENDIF
      \ENDFOR
      \RETURN \TRUE;
    \end{algorithmic}
  \end{algorithm}

%%%%%%%%%%%%%%%%%%%%%%%%%%%%%%%%%%%%%%%%%%%%%%%%%%%%%%%%%
\section{Message-Passing Decoding with ML Guarantee for \\Irregular LDPC Codes} \label{sec:NWMS}
%%%%%%%%%%%%%%%%%%%%%%%%%%%%%%%%%%%%%%%%%%%%%%%%%%%%%%%%%

In this section we present a weighted min-sum decoder (called, \NWMS)
for irregular Tanner codes with single parity-check local codes over any MBIOS
channel. In Section~\ref{sec:nwms proof} we prove that the decoder computes the ML codeword if
a locally-optimal codeword exists (Theorem~\ref{thm:NWMS}). Note that
Algorithm \NWMS\ is not presented as a min-sum process.  However, in
Section~\ref{sec:nwms proof}, an equivalent min-sum version is presented.

We deal with Tanner codes based on Tanner graphs
$G=(\calV\cup\calJ,E)$ with single parity-check local codes. Local-code
nodes $C \in \calJ$ in this case are called \emph{check nodes}.  The graph $G$ may be either regular or irregular.
All the results in this section hold for every Tanner graph,
regardless of its girth, degrees, or density.

A huge number of works deal with message-passing decoding algorithms. We point out three works that can be viewed as precursors to our decoding algorithm. Gallager~\cite{Gal63} presented the sum-product iterative decoding algorithm for
LDPC codes.  Tanner~\cite{Tan81} viewed iterative decoding algorithms
as message-passing iterative algorithms over the edges of the Tanner graph.
Wiberg~\cite{Wib96} characterized decoding failures of the min-sum
iterative decoding algorithm by negative cost trees.  Message-passing
decoding algorithms proceed by iterations of ``ping-pong'' messages
between the variable nodes and the local-code nodes in the Tanner
graph.  These messages are sent along the edges.

\paragraph{Algorithm description.}
Algorithm \NWMS$(\lambda,h,w)$, listed as
Algorithm~\ref{alg:weighted-min-sum}, is a normalized $w$-weighted
version of the min-sum decoding algorithm for decoding Tanner codes
with single parity-check local codes. The input to algorithm \NWMS\
consists of an LLR vector $\lambda\in \R^N$, an integer $h>0$ that
determines the number of iterations, and a nonnegative weight vector
$w\in \R_+^h\setminus\{0^h\}$.  For each edge $(v,C)$, each iteration consists of one
message from the variable node $v$ to the check node $C$ (that is, the
``ping'' message), and one message from $C$ to $v$ (that is, the
``pong'' message). Hence, the message complexity of
Algorithm~\ref{alg:weighted-min-sum} is $O(|E|\cdot h)$. (It can be
implemented so that the running time is also $O(|E|\cdot h)$).

Let $\mu_{v\rightarrow C}^{(l)}$ denote the ``ping'' message from a
variable node $v \in \calV$ to an adjacent check node $C \in \calJ$ in
iteration $l$ of the algorithm. Similarly, let $\mu_{C\rightarrow
  v}^{(l)}$ denote the ``pong'' message from $C \in \calJ$ to $v \in
\calV$ in iteration $l$.  Denote by $\mu_v$ the final value computed
by variable node $v \in \calV$. Note that the \NWMS\ decoding algorithm
does not add $w_0\lambda_v$ in the computation of $\mu_v$ in Line~11
for ease of presentation\footnote{Adding $w_0\lambda_v$ to $\mu_v$ in
  Line~11 requires changing the definition of \PNW\ deviations so that
  they also include the root of each $d$-tree.}.  The output of the
algorithm $\hat{x}\in \{0,1\}^N$ is computed locally by each variable
node in Line~12. In the case where $\mu_v=0$ we chose to assign $x_v=1$ for ease of presentation. However, one can choose to assign $x_v$ with either a `0' or a `1' with equal probability. Algorithm~\NWMS\ may be applied to any MBIOS channel
(e.g., BEC, BSC, AWGN, etc.) because the input is the LLR
vector\footnote{In the case of a BEC, the LLR vector $\lambda$ is in
  $\{+\infty,-\infty,0\}^N$.  In this case, all the messages in
  Algorithm~\ref{alg:weighted-min-sum} are in the set
  $\{-\infty,0,+\infty\}$. The arithmetic over this set is the
  arithmetic of the \emph{affinely extended real number system} (e.g., for a
  real $a$, $\pm\infty + a = \pm\infty$, etc.).  Under such
  arithmetic, there is no need to assign weights to the LLR value and
  the incoming messages in the computation of variable-to-check
  messages in Line~\ref{line:variable-to-check}. Notice that $+\infty$
  is never added to $-\infty$ since a BEC may only erase bits and can
  not flip any bit.  Therefore, all computed messages in
  Algorithm~\ref{alg:weighted-min-sum} are equal to either $\pm\infty$ or
  $0$.}.

\begin{algorithm}
\begin{algorithmic}[1]
\caption{\NWMS$(\lambda,h,w)$ - An iterative normalized weighted min-sum decoding algorithm.
Given an LLR vector $\lambda \in \R^N$ and level weights $w \in \R_+^{h}\setminus\{0^h\}$, outputs a binary string
$\hat{x} \in \{0,1\}^N$.
}
\label{alg:weighted-min-sum}
\STATE Initialize: $\forall C\in\calJ$, $\forall v\in\calN(C):$~~
$\mu_{C\rightarrow v}^{(-1)} \gets 0$
\FOR {$l=0$ to $h-1$}
\FORALL[``PING'']{$v\in\calV$, $C\in\calN(v)$}
\STATE $\mu_{v\rightarrow C}^{(l)} \gets \frac{w_{h-l}}{\deg_G(v)}\lambda_v+\frac{1}{\deg_G(v)-1} \sum_{C'\in\calN(v)\setminus\{C\}}\mu_{C'\rightarrow v}^{(l-1)}$ \label{line:variable-to-check}
\ENDFOR
\FORALL[``PONG'']{$C\in\calJ$, $v\in\calN(C)$}
\STATE $\mu_{C\rightarrow v}^{(l)} \gets \bigg(\prod_{u\in\calN(C)\setminus\{v\}}\textrm{sign}\big(\mu_{u\rightarrow C}^{(l)}\big)\bigg)\cdot\min\bigg\{\lvert\mu_{u\rightarrow C}^{(l)}\rvert~\bigg\vert~{u\in\calN(C)\setminus\{v\}}\bigg\}$ \label{line:check-to-variable}
\ENDFOR
\ENDFOR
\FORALL[Decision]{$v\in\calV$}
\STATE $\mu_v \gets \sum_{C\in\calN(v)}\mu_{C\rightarrow v}^{(h-1)}$
\STATE
$\hat{x}_v \gets \begin{cases} 0 & \mathrm{if\ }\mu_v>0,\\ 1 & \mathrm{otherwise}. \end{cases}$
\ENDFOR
\end{algorithmic}
\end{algorithm}

The upcoming Theorem~\ref{thm:NWMS} states that \NWMS$(\lambda,h,w)$ computes an
$(h,w,d\!=\!2)-$locally optimal codeword w.r.t. $\lambda$ if such a codeword
exists. Hence, Theorem~\ref{thm:NWMS} provides a sufficient condition for successful iterative decoding of the ML codeword for any finite number $h$ of iterations. In particular, the number of iterations may exceed (any function of) the girth.
Theorem~\ref{thm:NWMS} implies an alternative proof of the uniqueness of an $(h,w,d\!=\!2)$-locally optimal codeword that is proved in Theorem~\ref{thm:MLsufficient}. The proof appears
in Section~\ref{Sec:proofMPsufficient}.

\begin{theorem}[\NWMS\ decoding algorithm finds the locally optimal codeword] \label{thm:MPsufficient}\label{thm:NWMS}
Let $G=(\calV\cup\calJ,E)$ denote a Tanner graph and let $\calC(G)
\subset \{0,1\}^N$ denote the corresponding Tanner code with
single parity-check local codes. Let $h \in \mathds{N}_+$ and let $w \in
\R_+^h\setminus\{0^h\}$ denote a non-negative weight vector. Let $\lambda \in \R^N$
denote the LLR vector of the channel output. If $x \in \calC(G)$ is an
$(h,w,d\!=\!2)$-locally optimal codeword w.r.t. $\lambda$, then \NWMS$(\lambda,h,w)$ outputs $x$.
\end{theorem}

The message-passing algorithm \VERIFY\ (Algorithm~\ref{alg:verify}) described in
Section~\ref{subsec:certificateVerifiaction} can be used to verify
whether \NWMS$(\lambda,h,w)$ outputs the $(h,w,d\!=\!2)$-locally optimal
codeword w.r.t. $\lambda$.
If there exists $(h,w,d\!=\!2)$-locally optimal
codeword w.r.t. $\lambda$, then, by Theorem~\ref{thm:MLsufficient} and Theorem~\ref{thm:NWMS}, it holds that:
\begin{inparaenum}[(i)]
\item the output of \NWMS$(\lambda,h,w)$ is the unique ML codeword, and
\item algorithm \VERIFY\ returns \emph{true} for the decoded codeword.
\end{inparaenum}
If no $(h,w,d\!=\!2)$-locally optimal
codeword exists w.r.t. $\lambda$, then algorithm \VERIFY\ returns \emph{false} for every input codeword.
We can therefore obtain a message-passing decoding algorithm with an ML certificate obtained by local optimality by using Algorithms~\ref{alg:verify} and~\ref{alg:weighted-min-sum} as follows.

\ignore{Theorem~\ref{thm:MPsufficient} states that if an $(h,w,d\!=\!2)$-locally optimal codeword exists, then the \NWMS\ decoding algorithm is guaranteed to find it.}

Algorithm \certNWMS$(\lambda,h,w)$, listed as Algorithm~\ref{alg:certifiedNWMS}, is an ML-certified version of the \NWMS\ decoding algorithm.
The input to algorithm \certNWMS\ consists of an LLR vector $\lambda\in \R^N$, an integer $h>0$ that determines the number of iterations, and a nonnegative weight vector
$w\in \R_+^h\setminus\{0^h\}$.
If the \certNWMS\ decoding algorithm returns a binary word, then it is guaranteed to be the unique ML codeword w.r.t. $\lambda$. Otherwise, \certNWMS\ declares a failure to output an ML-certified codeword. The message complexity of
Algorithm~\ref{alg:weighted-min-sum} is $O(|E|\cdot h)$. (It can be
implemented so that the running time is also $O(|E|\cdot h)$).

\begin{algorithm}
\begin{algorithmic}[1]
\caption{\certNWMS$(\lambda,h,w)$ - An iterative normalized weighted min-sum decoding algorithm with an ML-certified output based on local optimality.
Given an LLR vector $\lambda \in \R^N$ and level weights $w \in \R_+^{h}\setminus\{0^h\}$, outputs the ML codeword $\hat{x} \in \{0,1\}^N$ w.r.t. $\lambda$ or a ``failure''.
}
\label{alg:certifiedNWMS}
\STATE $x \gets \NWMS(\lambda,h,w)$
\IF{$x$ is a codeword}
\IF[$x$ is $(h,w,2)$-locally optimal w.r.t. $\lambda$]{$\VERIFY(x,\lambda,h,w,2)=\TRUE$}
      \RETURN $x$;
\ENDIF
\ENDIF
\RETURN failure;
\end{algorithmic}
\end{algorithm}

\begin{remark}
Local optimality is a sufficient condition for ML.
In case that there is no $(h,w,d\!=\!2)$-locally optimal codeword w.r.t. $\lambda$, then the binary word that the \NWMS\ decoding algorithm outputs may be an ML codeword. Note however, that the \certNWMS\ decoding algorithm declares a failure to output an ML-certified codeword in this case.
In the case where a locally-optimal codeword exists, then both \NWMS\ decoding algorithm and \certNWMS\ decoding algorithm are guaranteed to output this codeword, which is the unique ML codeword w.r.t. $\lambda$.
\end{remark}

%%%%%%%%%%%%%%%%%%%%%%%%%%%%%%%%%%%%%%%%%%%%%%%%%%%%%%%%%%%%%%%%%%%%%%%%%%
\subsection{Symmetry of \NWMS\ Decoding Algorithm and the All-Zero Codeword Assumption}\label{subsec:symmetry}
%%%%%%%%%%%%%%%%%%%%%%%%%%%%%%%%%%%%%%%%%%%%%%%%%%%%%%%%%%%%%%%%%%%%%%%%%%

We define symmetric decoding algorithms (see~\cite[Definition 4.81]{RU08} for
a discussion of symmetry in message passing algorithms).

\begin{definition}[symmetry of decoding algorithm] \label{def:symmetricDec}
Let $x\in\calC$ denote a codeword and let $(-1)^x\in\{\pm1\}^N$ denote the vector whose $i$th component equals $(-1)^{x_i}$. Let $\lambda$ denote an LLR vector. A decoding algorithm, \DEC$(\lambda)$, is \emph{symmetric} w.r.t. code $\calC$, if
\begin{equation}
\forall x\in\calC.\ \ x\oplus\DEC(\lambda) = \DEC\big((-1)^x\ast\lambda\big).
\end{equation}
\end{definition}

\noindent
The following lemma states that the \NWMS\ decoding algorithm is symmetric. The
proof is by induction on the number of iterations.
\begin{lemma}[symmetry of \NWMS] \label{lemma:NWMSsymmetry}
Fix $h\in \N_+$ and $w \in \R_+^N$. Consider $\lambda \in \R^N$ and a codeword $x \in \calC(G)$. Then,
\begin{equation}
x\oplus \NWMS(\lambda,h,w) = \NWMS\big((-1)^x\ast\lambda,h,w\big).
\end{equation}
\end{lemma}
\begin{proof}
See Appendix~\ref{app:NWMSSymmetryProof}.
\end{proof}

\noindent The following corollary follows from Lemma~\ref{lemma:NWMSsymmetry}
and the symmetry of an MBIOS channel.
\begin{corollary}[All-zero codeword assumption] \label{corollary:NWMSzeroAssumption}
Fix $h\in \N_+$ and $w \in \R_+^N$. For MBIOS channels, the probability that the \NWMS\ decoding algorithm fails to decode the transmitted codeword is independent of the transmitted codeword itself. That is,
\[\Pr\{\mathrm{\NWMS~fails}\} = \Pr\big\{\NWMS(\lambda,h,w)\neq 0^N~\vert~c=0^N\}.\]
\end{corollary}
\begin{proof}
Following Lemma~\ref{lemma:NWMSsymmetry}, for every codeword $x$,
\begin{align*}
\Pr\big\{\NWMS&(\lambda,h,w)\neq x~\big\vert~c=x\big\}=\Pr\big\{\NWMS\big((-1)^x\ast\lambda,h,w\big)\neq 0^N~\big\vert~c=x\big\}.
\end{align*}
For MBIOS channels, $f(\lambda_i\mid c_i=0)=f(-\lambda_i\mid c_i=1)$. Therefore, the mapping $(x,\lambda)\mapsto(0^N,(-1)^x\ast\lambda)$ preserves the probability measure. We apply this mapping to $(x,(-1)^x\ast\lambda)\mapsto(0^N,(-1)^x\ast(-1)^x\ast\lambda)$ and conclude that
\begin{align*}
\Pr\big\{\NWMS\big((-1)^x\ast&\lambda,h,w\big)\neq 0^N~\big\vert~c=x\big\}=\Pr\big\{\NWMS(\lambda,h,w)\neq 0^N~\big\vert~c=0^N\big\}.
\end{align*}
\end{proof}

Following the contra-positive of Theorem~\ref{thm:MPsufficient} and Corollary~\ref{corollary:NWMSzeroAssumption}, provided that the channel is symmetric, for a fixed $h$ and $w \in \R_+^h\backslash\{0^h\}$, we have
\begin{align}\label{eqn:MPfailureBound}
\Pr\bigg\{\mathrm{NWMS}(\lambda,h,w)~\mathrm{fails}\bigg\}\leqslant\Pr\bigg\{\exists \beta \in \mathcal{B}_2^{(w)}~\mathrm{s.t.}~\langle \lambda, \beta \rangle\leqslant0~\bigg\vert~c=0^N\bigg\}.
\end{align}

Bounds on the existence of a non-positive \PNW\ deviation (i.e., the right-hand side in Equation~(\ref{eqn:MPfailureBound})) are discussed in Section~\ref{subsec:DiscussionNWMS}.

%%%%%%%%%%%%%%%%%%%%%%%%%%%%%%%%%%%%%%%%%%%%%%%%%%%%%
\subsection{Proof of Theorem~\ref{thm:MPsufficient} -- \NWMS\ Decoding Algorithm Finds the Locally Optimal Codeword}
\label{Sec:proofMPsufficient} \label{sec:nwms proof}
%%%%%%%%%%%%%%%%%%%%%%%%%%%%%%%%%%%%%%%%%%%%%%%%%%%%%

%\paragraph{Proof outline.}
\begin{proof}[Proof outline]
%The proof of Theorem~\ref{thm:MPsufficient} is structured as follows.
The proof of Theorem~\ref{thm:MPsufficient} is based on two observations.
%\begin{inparaenum}[(1)]
\begin{compactenum}[(i)]
\item We present an equivalent algorithm, called \NWMS2 (Section~\ref{subsec:NWMS2}), and prove that Algorithm \NWMS2 outputs the all-zero codeword if $0^N$ is locally optimal (Sections~\ref{subsec:NWMS2-DP}--\ref{subsec:LO}).
\item In Lemma~\ref{lemma:NWMSsymmetry} we proved that the \NWMS\ decoding algorithm is symmetric. This symmetry is w.r.t. the mapping of a pair $(x,\lambda)$ of a codeword and an LLR vector to a pair $(0^N,\lambda^0)$ of the all-zero codeword and a corresponding LLR vector $\lambda^0\triangleq(-1)^{x}\ast\lambda$ (recall that ``$\ast$'' denotes a coordinate-wise vector multiplication).
\end{compactenum}
%\end{inparaenum}

To prove Theorem~\ref{thm:MPsufficient}, we prove the contrapositive statement, that is, if $x\neq\NWMS(\lambda,h,w)$, then $x$ is not $(h,w,d\!=\!2)$-locally optimal w.r.t. $\lambda$.
Let $x$ denote a codeword, and let $(-1)^x$ denote the vector whose $i$th component equals $(-1)^{x_i}$. Define $\lambda^0\triangleq (-1)^x\ast\lambda$. By definition $\lambda = (-1)^x\ast\lambda^0$.

The proof is obtained by the following derivations.
Because $x\neq\NWMS(\lambda,h,w)$, it follows by Lemma~\ref{lemma:NWMSsymmetry} (symmetry of \NWMS) that
$x\neq x\oplus\NWMS(\lambda^0,h,w)$, and hence $0^N\neq\NWMS(\lambda^0,h,w)$. By the upcoming Lemma~\ref{lemma:negDevExists}, $0^N$ is not $(h,w,2)$-locally optimal w.r.t. $\lambda^0$. Because $\lambda = (-1)^x\ast\lambda^0$, it follows by Proposition~\ref{proposition:LOsymmetry} that $x$ is not $(h,w,2)$-locally optimal w.r.t. $\lambda$ as required.
\end{proof}
We are left to prove Lemma~\ref{lemma:negDevExists} used in the foregoing proof.

\subsubsection{NWMS2 : An Equivalent Version} \label{subsec:NWMS2}

The input  to Algorithm \NWMS\ includes
the LLR vector $\lambda$. We refer to this algorithm as a min-sum
decoding algorithm in light of the general description of Wiberg~\cite{Wib96} in the log-domain.
In Wiberg's description, every check node finds a minimum value from a
set of functions on the incoming messages, and every variable node
computes the sum of the incoming messages and its corresponding
channel observation. Hence the name min-sum.

Let $y\in \R^N$ denote channel observations. For $a\in \{0,1\}$,
define the log-likelihood of $y_i$ by $\lambda_{i}(a) \triangleq
-\log\big(f(y_i|c_i=a)\big)$. Note that the log-likelihood ratio
$\lambda_i$ for $y_i$ equals $\lambda_i(1) - \lambda_i(0)$.
For $a\in\{0,1\}$, let $\lambda(a)\in\R^N$ denote the log-likelihood vector whose $i$th component equals $\lambda_i(a)$.

Algorithm \NWMS2$(\lambda(0),\lambda(1),h,w)$, listed as
Algorithm~\ref{alg:weighted-min-sum2}, is a normalized $w$-weighted
min-sum decoding algorithm.  Algorithm \NWMS2 computes separate
reliabilities for ``$0$'' and ``$1$''. Namely, $\mu_{v\rightarrow
  C}^{(l)}(a)$ and $\mu_{C\rightarrow v}^{(l)}(a)$ denote the messages
corresponding to the assumption that node $v$ is assigned the value
$a$ (for $a\in\{0,1\}$). The higher the values of these messages, the lower the likelihood of the event $x_v=a$.

The main difference between the presentations of Algorithm~\ref{alg:weighted-min-sum} and Algorithm~\ref{alg:weighted-min-sum2} is in Line~7.
Consider a check node $C$ and valid assignment $x\in\{0,1\}^{\deg(C)}$ to variable nodes adjacent to $C$ with even weight.
For every such assignment $x$ in which $x_v=a$, the check node $C$ computes the sum of the incoming messages $\mu_{u\rightarrow C}^{(l)}(x_u)$ from the neighboring nodes $u\in\calN(C)\setminus\{v\}$. The message $\mu_{C\rightarrow v}^{(l)}(a)$ equals the minimum value over these valid summations.

\begin{algorithm}
\begin{algorithmic}[1]
\caption{NWMS2$(\lambda(0),\lambda(1),h,w)$ - An iterative normalized weighted min-sum decoding algorithm.
Given log-likelihood vectors $\lambda(a) \in \R^N$ for $a\in\{0,1\}$ and level weights $w \in \R_+^{h}\setminus\{0^h\}$, outputs a binary string
$\hat{x} \in \{0,1\}^N$.
}
\label{alg:weighted-min-sum2}
\STATE Initialize: $\forall C\in\calJ$, $\forall v\in\calN(C)$, $\forall a\in\{0,1\}:$~~~$\mu_{C\rightarrow v}^{(-1)}(a) \gets 0$
\FOR {$l=0$ to $h-1$}
\FORALL[``PING'']{$v\in\calV$, $C\in\calN(v)$, $a\in\{0,1\}$}
\STATE $\mu_{v\rightarrow C}^{(l)}(a) \gets \frac{w_{h-l}}{\deg_G(v)}\lambda_v(a)+\frac{1}{\deg_G(v)-1} \sum_{C'\in\calN(v)\setminus\{C\}}\mu_{C'\rightarrow v}^{(l-1)}(a)$
\ENDFOR
\FORALL[``PONG'']{$C\in\calJ$, $v\in\calN(C)$, $a\in\{0,1\}$}
\STATE \begin{align*}
\mu_{C\rightarrow v}^{(l)}(a) \gets \min
\left\{\sum_{u\in\calN(C)\setminus\{v\}}\mu_{u\rightarrow C}^{(l)}(x_{u})
\,\left|\,
\begin{aligned}
&x\in\{0,1\}^{\deg(C)}\\
&\lVert x\rVert_1\mathrm{~is~even}\\
&x_v = a\end{aligned}\right.
\right\}
\end{align*}
\ENDFOR
\ENDFOR
\FORALL[Decision]{$v\in\calV$}
\STATE $\mu_v(a) \gets \sum_{C\in\calN(v)}\mu_{C\rightarrow v}^{(h-1)}(a)$
\STATE
$\hat{x}_v \gets \begin{cases} 0 & \mathrm{if\ }\big(\mu_v(1)-\mu_v(0)\big) >0,\\ 1 & \mathrm{otherwise}. \end{cases}$
\ENDFOR
\end{algorithmic}
\end{algorithm}

Following Wiberg~\cite[Appendix A.3]{Wib96}, we claim that
Algorithms~\ref{alg:weighted-min-sum} and~\ref{alg:weighted-min-sum2}
are equivalent.

\begin{claim} \label{claim:NWMS-eq} Let $\lambda$, $\lambda(0)$, and
  $\lambda(1)$ in $\R^N$ denote the LLR vector and the two
  log-likelihood vectors for a channel output $y\in\R^N$. Then, for
  every $h\in\N_+$ and $w\in\R_+^h$, the following equalities hold:
  \begin{enumerate}
  \item $\mu_{v\rightarrow C}^{(l)} = \mu_{v\rightarrow C}^{(l)}(1) -
    \mu_{v\rightarrow C}^{(l)}(0)$ and $\mu_{C\rightarrow v}^{(l)} =
    \mu_{C\rightarrow v}^{(l)}(1) - \mu_{C\rightarrow v}^{(l)}(0)$ in
    every iteration $l$.
  \item $\mu_v = \mu_v(1)-\mu_v(0)$. Hence $\NWMS(\lambda,h,w)$ and
    $\NWMS2(\lambda(0),\lambda(1),h,w)$ output the same vector
    $\hat{x}$.
  \end{enumerate}
\end{claim}

\subsubsection{\NWMS2 as a Dynamic Programming Algorithm}
\label{subsec:NWMS2-DP}
In Lemma~\ref{lemma:NWMS2dp} we prove that
Algorithm \NWMS2 is a dynamic programming algorithm that computes, for
every variable node $v$, two min-weight valid configurations. One configuration is $0$-rooted and the other configuration is $1$-rooted.
Algorithm \NWMS2 decides $\hat{x}_v=0$ if the min-weight valid configuration rooted at $v$ is $0$-rooted, otherwise decides $\hat{x}_v=1$. We now
elaborate on the definition of valid configurations and their weight.

%\paragraph{Valid configurations and their weight.}
\emph{Valid configurations and their weight.}~~
Fix a variable node $r\in \cal V$. We refer to $r$ as the \emph{root}.
Consider the path-prefix tree $\calT_r^{2h}(G)$ rooted at $r$ consisting
of all the paths of length at most $2h$ starting at $r$.  Denote the vertices
of $\calT_r^{2h}$ by $\hat{\calV}\cup\hat{\calJ}$, where paths in
$\hat{\calV}=\{p\ \vert\ p\in\hat{V},\ t(p)\in\calV\}$ are
variable paths, and paths in $\hat{\calJ}=\{p\ \vert\ p\in\hat{V},\
t(p)\in\calJ\}$ are parity-check paths.
Denote by $(r)$ the zero-length path, i.e., the path consisting of only the root $r$.

A binary word $z \in \{0,1\}^{|\hat{\calV}|}$ is interpreted as an assignment to
variable paths $p\in\hat{\calV}$ where $z_p$ is assigned to $p$. We
say that $z$ is a \emph{valid configuration} if it satisfies all
parity-check paths in $\hat{\calJ}$. Namely, for every check path
$q\in\hat\calJ$, the assignment to its neighbors has an even number of
ones. We denote the set of valid configurations of $\calT_r^{2h}$ by
$\config(\calT_r^{2h})$.

The weight $\calW_{\calT_r^{2h}}(z)$ of a valid configuration $z$ is
defined by weights $\calW_{\calT_r^{2h}}(p)$ that are assigned to
variable paths $p\in \hat\calV$ as follows. We start with level
weights $w = (w_1,\ldots,w_h) \in \R_+^h$ that are assigned to levels
of variable paths in ${\calT_r^{2h}}$.
Define the weight of a variable path $p\in \hat\calV$ w.r.t.
$w$ by\footnote{We use the same notation as in Definition~\ref{def:weightedSubtree}.}
\[
\calW_{\calT_r^{2h}}(p)
\triangleq \frac{w_{\lvert p\rvert/2}}{\deg_G\big(t(p)\big)}\cdot \prod_{q\in\PPrefix(p)\cap\hat{\calV}}\frac{1}{\deg_G\big(t(q)\big)-1}.
\]
There is a difference between Definition~\ref{def:weightedSubtree} and $\calW_{\calT_r^{2h}}(p)$. A minor difference is that we do not divide by $\lVert w\rVert_1$ as in Definition~\ref{def:weightedSubtree}. The main difference is that in Definition~\ref{def:weightedSubtree} the product is taken over all paths in $\PPrefix(p)$. However, in $\calW_{\calT_r^{2h}}(p)$ the product is taken only over variable paths in $\PPrefix(p)$.

The \emph{weight of a valid configuration $z\in\{0,1\}^{\lvert\hat\calV\rvert}$} is defined by
\[
\calW_{\calT_r^{2h}}(z) \triangleq \sum_{p\in\hat\calV\setminus\{(r)\}} \lambda_{t(p)}(z_p)\cdot \calW_{\calT_r^{2h}}(p).
\]

Given a variable node $r\in\calV$ and a bit $a\in\{0,1\}$, our goal is to compute the value of a min-weight valid configuration $\calW^{\min}(r,a)$ defined by
\begin{align*}
  \calW^{\min}(\calT_r^{2h},a) \triangleq \min \bigg\{ \calW_{\calT_r^{2h}}(z) \bigg\vert~\begin{aligned}&z\in \config(\calT_r^{2h}),\\ &z_{(r)}=a\end{aligned}\bigg\}.
\end{align*}

In the following lemma we show that \NWMS2 computes $\calW^{\min}(\calT_r^{2h},a)$ for every $r\in\calV$ and $a\in\{0,1\}$. The proof is based on interpreting \NWMS2 as dynamic programming. See Appendix~\ref{app:NWMS2_DP} for details.

\begin{lemma}\label{lemma:NWMS2dp}
Consider an execution of $\NWMS2(\lambda(0),\lambda(1),h,w)$. For every variable node $r$, $\mu_r(a) = \calW^{\min}(\calT_r^{2h},a)$.
\end{lemma}

From Line~12 in Algorithm \NWMS2 we obtain the following corollary that characterizes \NWMS2 as a computation of
min-weight configurations.
\begin{corollary}\label{cor:NWMS2dp}
Let $\hat x$ denote the output of
  \NWMS2$(\lambda(0),\lambda(1),h,w)$.
  For every variable node $r$, \[\hat{x}_r =\begin{cases} 0 & \mathrm{if\ }\calW^{\min}(\calT_r^{2h},1)>\calW^{\min}(\calT_r^{2h},0),\\ 1& \mathrm{otherwise.}\end{cases}\]
\end{corollary}

Define the ${\mathcal W}^*$ cost of a configuration $z$ in $\calT_r^{2h}$ to be
\[
\calW^*_{\calT_r^{2h}}(z) \triangleq \sum_{p\in \hat\calV}
\lambda_{t(p)} \cdot {\mathcal W}_{\calT_r^{2h}} (p) \cdot z_p.\]
Note that $\mathcal{W}^*_{\calT_r^{2h}}(z)$ uses the LLR vector $\lambda$ (i.e., $\lambda_v
=\lambda_v(1) - \lambda_v(0)$).

\begin{corollary}\label{cor:NWMSdp}
  Let $\hat x$ denote the output of \NWMS$(\lambda,h,w)$.  Let
  $z^*$ denote a valid configuration in $\calT^{2h}_r$ with minimum
  $\mathcal W^*$ cost. Then, $\hat{x}_r =
  z^*_{(r)}$.
\end{corollary}
\begin{proof}
The derivation in Equation~(\ref{eqn:WW1}) shows that the valid configuration $z^*$ that minimizes the $\calW^*$ cost also minimizes the $\calW$ cost.
\begin{align}
\nonumber \argmin_{z\in\config(\calT_r^{2h})} \calW_{\calT_r^{2h}}(z) &\stackrel{(\mathrm{a})}{=} \argmin_{z\in\config(\calT_r^{2h})} \big\{ \calW_{T_r^{2h}}(z)-\calW_{\calT_r^{2h}}(0^{|\hat\calV|})\big\} \\ \nonumber
  &\stackrel{(\mathrm{b})}{=}
  \argmin_{z\in\config(\calT_r^{2h})}\bigg\{  \!\!\sum_{\{p\in\hat{\calV}~\vert~z_p=1\}}\!\!\!\!\!\!\lambda_{t(p)}(1) \cdot \calW_{\calT_r^{2h}}(p)-\!\!\!\!\!\!  \sum_{\{p\in\hat{\calV}~\vert~z_p=1\}}\!\!\!\!\!\!\!\lambda_{t(p)}(0)\cdot\mathcal{W}_{\calT_r^{2h}}(p) \bigg\}\\
  \nonumber
  &\stackrel{(\mathrm{c})}{=} \argmin_{z\in\config(\calT_r^{2h})}
  \sum_{p\in \hat\calV} \lambda_{t(p)}\cdot {\mathcal W}_{\calT_r^{2h}}(p)\cdot z_p\\
  &= \argmin_{z\in\config(\calT_r^{2h})}\calW^*_{\calT_r^{2h}}(z). \label{eqn:WW1}
\end{align}
Equality ($\mathrm{a}$) relies on the fact that $\calW_{\calT_r^{2h}}(0^{|\hat\calV|})$ is a constant. The summands $\lambda_{t(p)}(z_p)\cdot\calW_{\calT_r^{2h}}(p)$ in $\calW_{\calT_r^{2h}}(z)$ with $z_p=0$ are reduced by the substraction of the same summands in $\calW_{\calT_r^{2h}}(0^{|\hat\calV|})$. This leaves in Equality ($\mathrm{b}$) only summands that correspond to bits $z_p=1$.
Equality ($\mathrm{c}$) is obtained by the LLR definition $\lambda_{t(p)} = \lambda_{t(p)}(1) - \lambda_{t(p)}(0)$.

Let $\hat{x} = \NWMS(\lambda,h,w)$ and $\hat{y} = \NWMS2(\lambda(0),\lambda(1),h,w)$.
By Corollary~\ref{cor:NWMS2dp} and Equation~(\ref{eqn:WW1}), $\hat{y}_r=z^*_{(r)}$. By Claim~\ref{claim:NWMS-eq}, $\hat{x}_r = \hat{y}_r$, and the corollary follows.
\end{proof}

\subsubsection{Connections to Local Optimality}\label{subsec:LO}
The following lemma states that the \NWMS\ decoding algorithm computes the all-zero codeword if $0^N$ is locally optimal.

\begin{lemma} \label{lemma:negDevExists}
Let $\hat{x}$ denote the output of $\NWMS(\lambda,h,w)$.
If $0^N$ is $(h,w,d\!=\!2)$-locally optimal w.r.t. $\lambda$, then $\hat{x}=0^N$.
\end{lemma}
\begin{proof}
We prove the contrapositive statement. Assume that $\hat{x}\neq0^N$. Hence, there exists a variable node $v$ for which $\hat{x}_v = 1$. Consider $\calT_v^{2h}=(\hat{\calV}\cup\hat{\calJ},\hat{E})$. Then, by Corollary~\ref{cor:NWMSdp}, there exists a valid configuration $z^*\in\{0,1\}^{|\hat{\calV}|}$ in $\calT_v^{2h}$ with $z^*_{(v)} = 1$ such that for every valid configuration $y\in \calT_v^{2h}$ it holds that
\begin{equation}\label{eqn:z*opt}
\calW^*_{\calT_v^{2h}}(z^*) \leqslant \calW^*_{\calT_v^{2h}}(y).
\end{equation}

Let $\calT(z^*)$ denote the subgraph of $\calT_v^{2h}$ induced by $\hat{\calV}(z^*)\cup\calN\big(\hat{\calV}(z^*)\big)$ where $\hat{\calV}(z^*) = \{p\in\hat{\calV}~\mid~z^*_p=1\}$. Note that $\calT(z^*)$ is a forest.
Because $z^*_{(v)}=1$ and $z^*$ is a valid configuration in $\calT_v^{2h}$, the forest $\calT(z^*)$ must contain a $2$-tree of height $2h$ rooted at the node $v$; denote this tree by $\calT$. Let $\tau \in \{0,1\}^{|\hat{\calV}|}$ denote the characteristic vector of the support of $\calT$, and let $z^0\in \{0,1\}^{|\hat{\calV}|}$ denote the characteristic vector of the support of $\calT(z^*) \setminus \calT$.
Then, $z^* = \tau + z^0$, where $z^0$ is also necessarily a valid configuration.
By linearity and disjointness of $\tau$ and $z^0$, we have
\begin{equation}\label{eqn:decompose}
\calW^*_{\calT_v^{2h}}(z^*) = \calW^*_{\calT_v^{2h}}(\tau + z^0) = \calW^*_{\calT_v^{2h}}(\tau) + \calW^*_{\calT_v^{2h}}(z^0).
\end{equation}
Because $z^0$ is a valid configuration, by Equation~(\ref{eqn:z*opt}), we have $\calW^*_{\calT_v^{2h}}(z^*) \leqslant \calW^*_{\calT_v^{2h}}(z^0)$. By Equation~(\ref{eqn:decompose}), it holds that $\calW^*_{\calT_v^{2h}}(\tau) \leqslant 0$.

Let $w^*_\tau \in R^{|\hat\calV|}$ denote the vector whose component indexed by $p\in\hat\calV$ equals \mbox{$\calW_{\calT_v^{2h}}(p) \cdot \tau_{p}$}. The vector $w^*_\tau$ is equal to the $w$-weighted $2$-tree $w_\calT$ according to Definition~\ref{def:weightedSubtree}. Hence, $\beta = \frac{1}{c}\cdot\pi_{G,\calT,w} \in \calB_2^{(w)}$ satisfies $\langle \lambda,\beta\rangle = \frac{1}{c}\cdot \calW^*_{\calT_v^{2h}}(\tau) \leqslant 0$, where $c\geqslant1$ is a normalizing constant so that $\frac{1}{c}\cdot\pi_{G,\calT,w}\in[0,1]^N$ (see Definition~\ref{def:PNWdevSet}). We therefore conclude that $0^N$ is not $(h,w,d\!=\!2)$-locally optimal w.r.t. $\lambda$ and the lemma follows.
\end{proof}

%%%%%%%%%%%%%%%%%%%%%%%%%%%%%%%%%%%%%%%%%%%%%%%%%%%%%%%%%%%%%%%%%%%%%%%%%%%%%%%%%%
\subsection{Numerical Results for Regular LDPC Codes} \label{subsec:results}
%%%%%%%%%%%%%%%%%%%%%%%%%%%%%%%%%%%%%%%%%%%%%%%%%%%%%%%%%%%%%%%%%%%%%%%%%%%%%%%%%%

We chose a $(3,6)$-regular LDPC code with block length $N=4896$ for which the girth of the Tanner graph equals $12$~\cite{RV00}. We ran up to $h=400$ iterations of the \NWMS\ decoding algorithm (Algorithm~\ref{alg:weighted-min-sum}) and the \certNWMS\ decoding algorithm (Algorithm~\ref{alg:certifiedNWMS}) for received words over an AWGN channel. Three choices of level weights $w$ were
considered: (1)~Unit level weights, $w_\ell\triangleq1$.  This choice
reduces local optimality to~\cite{ADS09,HE11} (although in these
papers $h$ is limited by a quarter of the girth). (2)~Geometric level
weights $w_\ell\triangleq3\cdot(3-1)^{\ell-1}=3\cdot2^{\ell-1}$. In this case the \NWMS\ decoding algorithm
reduces to the standard min-sum decoding algorithm~\cite{Wib96}. (3)~Geometric
level weights $w_\ell\triangleq3\cdot(\frac{3-1}{1.25})^{\ell-1}=3\cdot(\frac{2}{1.25})^{\ell-1}$.  In this case the
\NWMS\ decoding algorithm reduces to normalized BP-based algorithm with $\alpha =
1.25$~\cite{CF02}. The choice of weights in~\cite{CF02} was obtained by optimizing
density evolution w.r.t. minimum bit error probability.

\ignore{We chose a $(3,6)$-regular LDPC code with block length $N=4896$ for which the girth of the Tanner graph equals $12$~\cite{RV00}. We ran up to $h=400$ iterations of the \NWMS\ decoding algorithm for received words over an AWGN channel. The decoded words were checked by the local
optimality verifier (Algorithm~\ref{alg:verify}). Three choices of level weights $w$ were
considered: (1)~Unit level weights, $w_\ell\triangleq1$.  This choice
reduces local optimality to~\cite{ADS09,HE11} (although in these
papers $h$ is limited by a quarter of the girth). (2)~Geometric level
weights $w_\ell\triangleq3\cdot2^{\ell-1}$. In this case the \NWMS\ decoding algorithm
reduces to the standard min-sum decoding algorithm~\cite{Wib96}. (3)~Geometric
level weights $w_\ell\triangleq3\cdot1.25^{\ell-1}$.  In this case the
\NWMS\ decoding algorithm reduces to normalized BP-based algorithm with $\alpha =
1.25$~\cite{CF02}. The choice of weights in~\cite{CF02} was obtained by optimizing
density evolution w.r.t. minimum bit error probability.
}

Figure~\ref{fig:FERcurve} depicts the word error rate of the \NWMS\ decoding algorithm with respect to these three level weights by solid lines.
The word error rate of the \certNWMS\ decoding algorithm with respect to these three level weights is depicted by dashed lines, i.e., the dashed lines depict the cases in which the \NWMS\ decoding algorithm failed to return the transmitted codeword
certified as locally optimal (and hence ML optimal).
The error rate of LP decoding and sum-product decoding algorithm are depicted as well for comparison.

The results show that the choice of unit level weights
minimizes the gap between the cases in which the \NWMS\ decoding algorithm fails to decode the transmitted codeword with and without an ML-certificate by local optimality.
That is, with unit level weights the main cause for a failure in decoding the transmitted codeword is the lack of a locally optimal codeword.
Moreover, there is a tradeoff between maximizing the rate of successful ML-certified decoding by the \certNWMS\ decoding algorithm, and minimizing the (not necessarily ML-certified) word error rate by the \NWMS\ decoding algorithm.
This tradeoff was also observed in~\cite{JP11}.
\ignore{The gap between the cases in which the transmitted codeword is not locally optimal even though that the transmitted codeword is successfully decoded by the weighted min-sum algorithm is observed also in~\cite{JP11} in the case of another weight vector.
}

\begin{figure}[ht]
  \begin{center}
 \includegraphics[width=\textwidth]{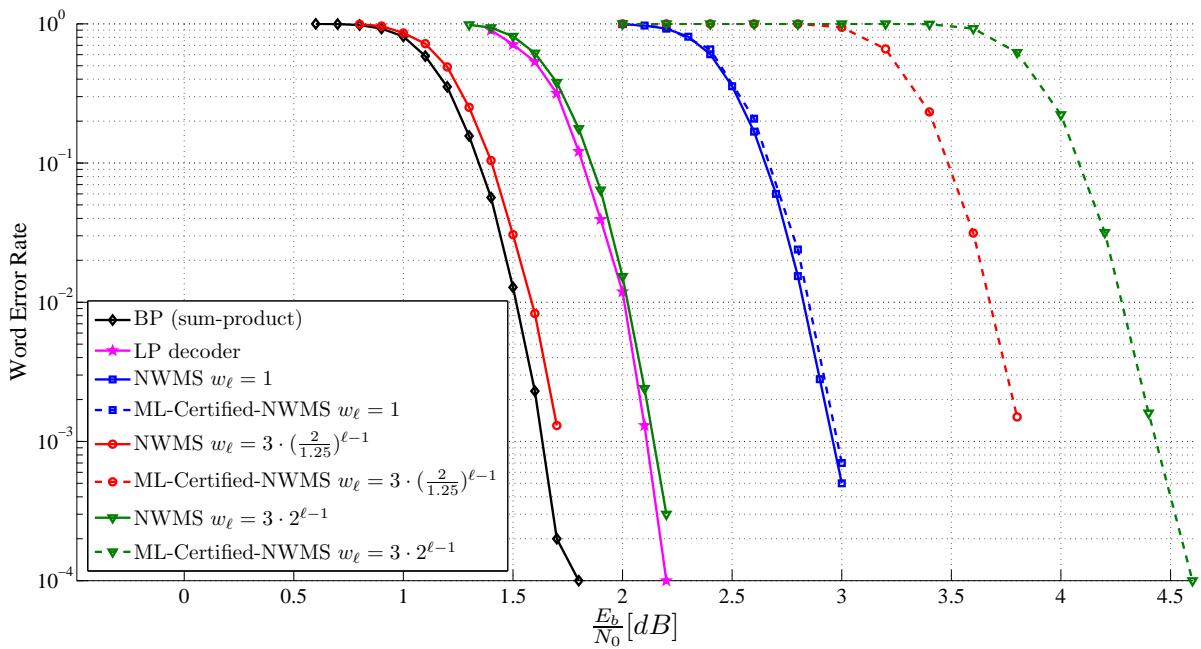}
 \caption{ Simulations for $(3,6)$-regular LDPC code of length
   $N=4896$~\cite{RV00} over a BI-AWGN channel.  Solid lines depict the WER of the
   \NWMS\ decoding algorithm for three level weights, the sum-product decoding algorithm and LP decoding. Dashed lines depict the word error rate of the
   \certNWMS\ decoding algorithm, i.e., the probability that the transmitted codeword is not locally optimal.}
  \label{fig:FERcurve}
  \end{center}
\end{figure}

%%%%%%%%%%%%%%%%%%%%%%%%%%%%%%%%%%%%%%%%%%%%%%%%%%%%%%%%%%%%%%%
\section{Bounds on the Error Probability of LP Decoding Using\\ Local Optimality} \label{sec:bounds}
%%%%%%%%%%%%%%%%%%%%%%%%%%%%%%%%%%%%%%%%%%%%%%%%%%%%%%%%%%%%%%%

In this section we analyze the probability that a local optimality certificate for regular Tanner codes exists, and therefore LP decoding succeeds. The analysis is based on the study of a sum-min-sum process that characterizes $d$-trees of a regular Tanner graph. We prove upper bounds on the error probability of LP decoding of regular Tanner codes in MBIOS channels. The upper bounds on the error probability imply lower bounds on the noise threshold of LP decoding for channels in which the channel parameter increases with noise level (e.g., BSC($p$) and BI-AWGNC($\sigma$))\footnote{On the other hand, upper bounds on the error probability imply upper bounds on the channel parameter threshold for channels in which the channel parameter is inverse proportional to the noise level (e.g., BI-AWGN channel described with a channel parameter of signal-to-noise ratio $\frac{E_b}{N_0}$).}. We apply the analysis to a BSC, and compare our results with previous results on expander codes.
The analysis presented in this section generalizes the probabilistic analysis of Arora \emph{et al.}~\cite{ADS09} from $2$-trees (skinny trees) to $d$-trees for any $d \geqslant 2$.

In the remainder of this section, we restrict our discussion to $(\dL,\dR)$-regular Tanner codes with minimum local distance $d^*$. Let $d$ denote a parameter such that $2 \leqslant d \leqslant d^*$.

The upcoming Theorem~\ref{thm:main-bound-BSC} summarizes the main results presented in this section for a BSC, and generalizes to any MBIOS channel as described in Section~\ref{subsec:MBIOSbound}.
Concrete bounds are given for a $(2,16)$-regular Tanner code with code rate at least $0.375$ when using $[16,11,4]$-extended Hamming codes as local codes.

\begin{theorem} \label{thm:main-bound-BSC} Let $G$ denote a
  $(\dL,\dR)$-regular bipartite graph with girth $g$, and let
  $\calC(G)\subset\{0,1\}^N$ denote a Tanner code based on $G$ with
  minimum local distance $d^*$. Let $x\in\calC(G)$ be a codeword.
  Suppose that $y\in\{0,1\}^N$ is obtained from $x$ by a BSC with
  crossover probability $p$. Then,
\begin{enumerate}
\item{[finite length bound]} Let $d=d_0$, $p\leqslant p_0$, $(\dL,\dR) = (2,16)$, and $d^*=4$. For the values of $d_0$ and $p_0$ in the rows labeled ``finite'' in Table~\ref{table:bounds} it holds that $x$ is the unique optimal solution to the LP decoder with probability at least
    \begin{align}\label{eqn:LP-finiteLB}
    \Pr\big\{\hat{x}^{\LP}(y)=x\big\} \geqslant 1-N\cdot {\alpha^{(d-1)}}^{\lfloor\frac{1}{4}g\rfloor}
    \end{align}
    for some constant $\alpha<1$.
\item{[asymptotic bound]} Let $d=d_0$, $(\dL,\dR) = (2,16)$, $d^*=4$, and $g=\Omega(\log N)$ sufficiently large. For the values of $d_0$ and $p_0$ in the rows labeled ``asymptotic'' in Table~\ref{table:bounds} it holds that $x$ is the unique optimal solution to the LP decoder with probability at least $1-\exp(-N^{\delta})$ for some constant $0<\delta<1$, provided that $p \leqslant p_0(d_0)$.
\item Let $d'\triangleq d-1$, $\dL'\triangleq \dL-1$, and $d'\triangleq \dR-1$. For any $(\dL,\dR)$ and $2 \leqslant d \leqslant d^*$ s.t. $\dL'\cdot d' \geqslant 2$, the codeword $x$ is the unique optimal solution to the LP decoder with probability at least $1-N\cdot {\alpha^{(\dL'\cdot d')}}^{\lfloor\frac{1}{4}g\rfloor}$ for some constant $\alpha<1$, provided that
\begin{align*}
\min_{t\geqslant0}\bigg\{\alpha_1(p,d, \dL, \dR,t)\cdot\big( \alpha_2(p,d, \dL, \dR,t) \big)^{1/(\dL'\cdot d'-1)}\bigg\}<1,
\end{align*}
\indent where
\begin{align*}
\alpha_1(p,d, \dL, \dR,t) =& \sum_{k=0}^{d'-1}{{\binom{\dR'}{k} p^k(1-p)^{(\dR'-k)}} e^{-t(d'-2k)}}\\&+\bigg(\sum_{k=d'}^{\dR'}\binom{\dR'}{k} p^k(1-p)^{\dR'-k}\bigg)e^{td'},\\
  \alpha_2(p,d, \dL, \dR,t)=& \binom{\dR'}{d'}\big((1-p)e^{-t}+pe^t\big)^{d'}.
\end{align*}
\end{enumerate}
\end{theorem}

\begin{table}[!t]
\renewcommand{\arraystretch}{1.3}
\caption{Computed values of $p_0$ for finite $d_0<d^*$ in Theorem~\ref{thm:main-bound-BSC} with respect to a BSC.}\label{table:bounds}
\begin{center}
\begin{tabular}{|c|c|c|}
\hline
 & $d_0$ & $p_0$ \\
 \hline
 \hline
 \multirow{2}{*}{finite} & $3$ & $0.0086$ \\
 \cline{2-3}
 & $4$ & $0.0218$ \\
 \hline
 \hline
  \multirow{2}{*}{asymptotic} & $3$ & 0.019 \\
   \cline{2-3}
 & $4$ & 0.044 \\
 \hline
\end{tabular}
\end{center}
{\scriptsize Values are presented for $(2,16)$-Tanner code with rate at least $0.375$ when using $[16,11,4]$-extended Hamming codes as local codes. Values in rows labeled ``finite'' refer to a finite-length bound: $\forall p\leqslant p_0$ the probability that the LP decoder succeeds is lower bounded by a function of $d$ and the girth of the Tanner graph (see Equation~(\ref{eqn:LP-finiteLB})). Values in rows labeled ``asymptotic'' refer to an asymptotic bound: For $g=\Omega(\log N)$ sufficiently large, the LP decoder succeeds w.p. at least $1-\exp(-N^{\delta})$ for some constant $0<\delta<1$, provided that $p \leqslant p_0(d_0)$.}
\end{table}
\begin{proof}[Proof Outline]
Theorem~\ref{thm:main-bound-BSC} follows from Lemma~\ref{lemma:LPsuccessLoose}, Lemma~\ref{lemma:uniformBound}, Corollary~\ref{cor:improvedBoundExample3}, and Corollary~\ref{cor:improvedBoundExample4} as follows. Part~1, that states a finite-length result, follows from Lemma~\ref{lemma:LPsuccessLoose} and Corollaries~\ref{cor:improvedBoundExample3} and~\ref{cor:improvedBoundExample4} by taking $s = 0 < h < \frac{1}{4}\girth(G)$ which holds for any Tanner graph $G$. Part~2, that deals with an asymptotic result, follows from Lemma~\ref{lemma:LPsuccessLoose} and Corollaries~\ref{cor:improvedBoundExample3} and~\ref{cor:improvedBoundExample4} by fixing $s=10$ and taking $g=\Omega(\log N)$ sufficiently large such that $s<h=\Theta(\log N) < \frac{1}{4}\girth(G)$. It therefore provides a lower bound on the threshold of LP decoding. Part~3, that states a finite-length result for any $(\dL,\dR)$-regular LDPC code, follows from Lemma~\ref{lemma:LPsuccessLoose} and Lemma~\ref{lemma:uniformBound}.
\end{proof}

We refer the reader to Section~\ref{subsec:DiscussionBounds} for a discussion on the results stated in Theorem~\ref{thm:main-bound-BSC}.
We now provide more details and prove the lemmas and corollaries used in the proof of Theorem~\ref{thm:main-bound-BSC}.

In order to simplify the probabilistic analysis of algorithms for decoding linear codes over symmetric channels, we apply the assumption that the all-zero codeword is transmitted, i.e., $c = 0^N$. Note that the correctness of the all-zero assumption depends on the employed decoding algorithm. Although this assumption is trivial for ML decoding because of the symmetry of a linear code $\mathcal{C}(G)$, it is not immediately clear in the context of LP decoding. Feldman \emph{et al.}~\cite{Fel03, FWK05} noticed that the fundamental polytope $\mathcal{P}(G)$ of Tanner codes with single parity-check local codes is highly symmetric, and proved that for MBIOS channels, the probability that the LP decoder fails to decode the transmitted codeword is independent of the transmitted codeword. The symmetry property of the polytope remains also for the generalized fundamental polytope of Tanner codes based on non-trivial linear local codes. Therefore, one can assume that $c=0^N$ when analyzing failures of LP decoding to decode the transmitted codeword for linear Tanner codes.
The following corollary is the contrapositive statement of Theorem~\ref{thm:LPsufficient} given $c=0^N$.

\begin{corollary} \label{cor:LPfailure}
For every fixed $h\in\N$, $w \in \R_+^h\backslash\{0^h\}$, and $2 \leqslant d \leqslant d^*$,
\begin{align*}
\Pr\bigg\{\mathrm{LP~decoding~fails}\bigg\}
\leqslant \Pr\bigg\{&\exists \beta \in \mathcal{B}_d^{(w)}~\mathrm{s.t.}~\langle\lambda,\beta\rangle\leqslant0~\bigg\vert~c=0^N\bigg\}.
\end{align*}
\end{corollary}

\subsection{Bounding Processes on Trees}

Let $G$ be a $(\dL,\dR)$-regular Tanner graph, and fix
$h<\frac{1}{4}\girth(G)$. Let $\calT_{v_0}^{2h}(G)$ denote the
path-prefix tree rooted at a variable node $v_0$ with height $2h$.
Since $h<\frac{1}{4}\girth(G)$, it follows that the projection of
$\calT_{v_0}^{2h}(G)$ to $G$ is a tree. Usually one regards a
path-prefix tree as an out-branching, however, for our analysis it is
more convenient to view the path-prefix tree as an in-branching.
Namely, we direct the edges of $\calT_{v_0}^{2h}$ so that each path in
$\calT_{v_0}^{2h}$ is directed toward the root $v_0$.  For $l \in
\{0,\ldots,2h\}$, denote by $V_l$ the set of vertices of
$\calT_{v_0}^{2h}$ at height $l$ (the leaves have height $0$ and the
root has height $2h$).  Let $\tau \subseteq V(\calT_{v_0}^{2h})$
denote the vertex set of a $d$-tree rooted at $v_0$.

\begin{definition} [$(h,\omega, d)$-Process on a $(\dL,\dR)$-Tree] \label{def:process}
Let $\omega \in \R_+^h$ denote a weight vector.
Let $\lambda$ denote an assignment of real values to the variable nodes of $\calT_{v_0}$. We define the $\omega$-weighted value of a $d$-tree $\tau$ by
\[ \mathrm{val}_\omega (\tau;\lambda) \triangleq \sum_{l=0}^{h-1}\sum_{v \in \tau \cap V_{2l}}\omega_l \cdot \lambda_v. \]
Namely, the sum of the values of variable nodes in $\tau$ weighted according to their height.

Given a probability distribution over assignments $\lambda$, we are interested in the probability
\begin{equation*}
\Pi_{\lambda,d,\dL,\dR}(h,\omega) \triangleq \mathrm{Pr}_\lambda \bigg\{ \min_{\tau \in \calT[v_0,2h,d]} \mathrm{val}_\omega(\tau;\lambda)~\leqslant~0 \bigg\} .
\end{equation*}
\end{definition}
In other words, $\Pi_{\lambda,d,\dL,\dR}(h,\omega)$ is the probability that the minimum value over all $d$-trees of height $2h$ rooted in some variable node $v_0$ in a $(\dL,\dR)$-bipartite graph $G$ is non-positive.
For every two roots $v_0$ and $v_1$, the trees $\calT_{v_0}^{2h}$ and $\calT_{v_1}^{2h}$ are isomorphic, hence $\Pi_{\lambda,d,\dL,\dR}(h,\omega)$ does not depend on the root $v_0$.

With this notation, the following lemma connects between the $(h,\omega, d)$-process on $(\dL,\dR)$-trees and the event where the all-zero codeword is $(h,w,d)$-locally optimal. We apply a union bound utilizing Corollary~\ref{cor:LPfailure}, as follows.
\begin{lemma}\label{lemma:LPsuccessLoose}
Let $G$ be a $(\dL,\dR)$-regular bipartite graph and $w \in \R_+^h\setminus\{0^h\}$ be a weight vector with $h<\frac{1}{4}\girth(G)$. Assume that the all-zero codeword is transmitted, and let $\lambda \in \R^N$ denote the LLR vector received from the channel. Then, $0^N$ is $(h,w,d)$-locally optimal w.r.t. $\lambda$ with probability at least
\[ 1-N \cdot \Pi_{\lambda,d,\dL,\dR}(h,\omega),\]
where $\omega_l = w_{h-l}\cdot \dL^{-1}\cdot (\dL-1)^{l-h+1} \cdot (d-1)^{h-l}$,
and with at least the same probability, $0^N$ is also the unique optimal LP solution given $\lambda$.
\end{lemma}

Note the two different weight notations that we use for consistency with~\cite{ADS09}: (i)~$w$ denotes a weight vector in the context of $(h,w,d)$-local optimality certificate, and (ii)~$\omega$ denotes a weight vector in the context of $d$-trees in the $(h,\omega,d)$-process. A one-to-one correspondence between these two vectors is given by $\omega_l = w_{h-l}\cdot \dL^{-1}\cdot (\dL-1)^{l-h+1} \cdot (d-1)^{h-l}$ for $0\leqslant l < h$. From this point on, we will use only $\omega$ in this section.

Following Lemma~\ref{lemma:LPsuccessLoose}, it is sufficient to estimate the probability $\Pi_{\lambda,d,\dL,\dR}(h,\omega)$ for a given weight vector $\omega$, a distribution of a random vector $\lambda$, constant $2 \leqslant d \leqslant d^*$, and degrees $(\dL,\dR)$.
Arora \emph{et al.}~\cite{ADS09} introduced a recursion for estimating and bounding the probability of the existence of a $2$-tree (skinny tree) with non-positive value in a $(h,\omega,2)$-process. We generalize the recursion and its analysis to $d$-trees with $2 \leqslant d \leqslant d^*$.

For a set $S$ of real values, let $\min^{[i]}\{S\}$ denote the $i$th smallest member in $S$. Let $\{\gamma\}$ denote an ensemble of i.i.d.\ random variables. Define random variables $X_0 ,\ldots, X_{h-1}$ and $Y_0 , \ldots , Y_{h-1}$ with the following recursion:
\begin{align}
% \nonumber to remove numbering (before each equation)
  Y_0 &= \omega_0 \gamma \label{eqn:Y_0}\\
  X_l &= \sum_{i=1}^{d-1}\mathrm{min}^{[i]}\big\{Y_l^{(1)},\ldots,Y_l^{(\dR-1)}\big\} &(0 \leqslant l < h) \label{eqn:X_l}\\
  Y_l &= \omega_l \gamma +  X_{l-1}^{(1)} +\cdots+X_{l-1}^{(\dL-1)}&(0 < l < h) \label{eqn:Y_l}
\end{align}
The notation $X^{(1)},\ldots,X^{(k)}$ and $Y^{(1)},\ldots,Y^{(k)}$ denotes $k$ mutually independent copies of the random variables $X$ and $Y$, respectively.
Each instance of $Y_l$, $0 \leqslant l < h$, uses an independent instance of a random variable $\gamma$.
Note that for every $0 \leqslant l < h$, the $d-1$ order statistic random variables $\big\{\min^{[i]}\{Y_l^{(1)},\ldots,Y_l^{(\dR-1)}\}~\big\vert~1\leqslant i\leqslant d-1\big\}$ in Equation~(\ref{eqn:X_l}) are dependent.

Consider a directed tree $\calT=\calT_{v_0}$ of height $2h$, rooted at node $v_0$. Associate variable nodes of $\calT$ at height $2l$ with copies of $Y_l$, and check nodes at height $2l+1$ with copies of $X_l$, for $0 \leqslant l <h$. Note that any realization of the random variables $\{\gamma\}$ to variable nodes in $\calT$ can be viewed as an assignment $\lambda$. Thus, the minimum value of a $d$-tree of $\calT$ equals $\sum_{i=1}^{\dL}X_{h-1}^{(i)}$. This implies that the recursion in (\ref{eqn:Y_0})--(\ref{eqn:Y_l}) defines a dynamic programming algorithm for computing $\min_{\tau\in\calT[v_0,2h,d]} \mathrm{val}_\omega(\tau;\lambda)$. Now, let the components of the LLR vector $\lambda$ be i.i.d.\ random variables distributed identically to $\{\gamma$\}, then
\begin{equation}
\Pi_{\lambda,d,\dL,\dR}(h,\omega) = \Pr\bigg\{\sum_{i=1}^{\dL}X_{h-1}^{(i)} \leqslant 0\bigg\}.
\end{equation}

Given a distribution of $\{\gamma\}$ and a finite ``height'' $h$, the challenge is to compute the distribution of $X_l$ and $Y_l$ according to the recursion in (\ref{eqn:Y_0})--(\ref{eqn:Y_l}). The following two lemmas play a major role in proving bounds on
$\Pi_{\lambda,d,\dL,\dR}(h,\omega)$.

\begin{lemma}[\!\!\cite{ADS09}]\label{lemma:ADSbound1}
For every $t \geqslant 0$,
\[ \Pi_{\lambda,d,\dL,\dR}(h,\omega) \leqslant \big(\E e^{-tX_{h-1}}\big)^{\dL}.\]
\end{lemma}
Let $d' \triangleq d-1$, $\dL' \triangleq \dL-1$ and $\dR' \triangleq \dR-1$.
\begin{lemma}[following \cite{ADS09}]\label{lemma:ADSbound2}
 For $0\leqslant s < l < h$, we have
 \begin{align*}
 \E e^{-tX_l} \leqslant {{\bigg(\E e^{-tX_s} \bigg)}^{(\dL'\cdot d')}}^{l-s}\cdot \prod_{k=0}^{l-s-1}{\bigg(\binom{\dR'}{d'}\big(\E e^{-t\omega_{l-k}\gamma}\big)^{d'}\bigg)^{(\dL'\cdot d')}}^k.
\end{align*}
\end{lemma}
\begin{proof}
See Appendix~\ref{app:proofofADSbound2}.
\end{proof}

In the following subsection we present concrete bounds on  $\Pi_{\lambda,d,\dL,\dR}(h,\omega)$ for a BSC. The bounds are based on Lemmas~\ref{lemma:ADSbound1} and~\ref{lemma:ADSbound2}.
The technique used to derive concrete bounds for a BSC may be applied to other MBIOS channels. For example, concrete bounds for a BI-AWGN channel can be derived by a generalization of the analysis presented in~\cite{HE11}.

\subsection{Analysis for a Binary Symmetric Channel}

Consider a binary symmetric channel with crossover probability $p$ denoted by BSC($p$). In the case that the all-zero codeword is transmitted, the channel input is $c_i=0$ for every $i$. Hence, $\Pr\big\{\lambda_i = -\log\big(\frac{1-p}{p}\big)\big\}=p$, and $\Pr\big\{\lambda_i =+\log\big(\frac{1-p}{p}\big)\big\} = 1-p$.
Since $\Pi_{\lambda,d,\dL,\dR}(h,\omega)$ is invariant under positive scaling of the vector $\lambda$, we consider in the following analysis the scaled vector $\lambda$ in which $\lambda_i = +1$ w.p. $p$, and $-1$ w.p. $(1-p)$.

Following the analysis of Arora \emph{et al.}~\cite{ADS09}, we apply a simple analysis in the case of uniform weight vector $\omega$. Then, we present improved bounds by using a non-uniform weight vector.

\subsubsection{Uniform Weights}
Consider the case where $\omega = 1^h$. Let $\alpha_1 \triangleq \E e^{-tX_0}$ and $\alpha_2 \triangleq \binom{\dR'}{d'}(\E e^{-t\gamma} )^{d'}$ where $\gamma\stackrel{\mathrm{i.i.d.}}{\sim}\lambda_i$, and define $\alpha \triangleq \min_{t\geqslant 0}\alpha_1 \cdot \alpha_2^{1/(\dL'\cdot d'-1)}$. Note that $\alpha_1 \leqslant \alpha_2$ (see Equation~(\ref{eqn:2}) in Appendix~\ref{app:proofofADSbound2}).We consider the case where $\alpha<1$.
By substituting notations of $\alpha_1$ and $\alpha_2$ in Lemma~\ref{lemma:ADSbound2} for $s=0$, we have
\begin{align*}
\E e^{-tX_l}
&\leqslant  {{\bigg(\E e^{-tX_0} \bigg)}^{(\dL'\cdot d')}}^{l}\cdot \prod_{k=0}^{l-1}{\bigg(\binom{\dR'}{d'}\big(\E e^{-t\gamma}\big)^{d'}\bigg)^{(\dL'\cdot d')}}^k\\
&= {{\alpha_1}^{(\dL'\cdot d')}}^{l}\cdot \prod_{k=0}^{l-1}{{\alpha_2}^{(\dL'\cdot d')}}^k\\
&= {{\alpha_1}^{(\dL'\cdot d')}}^{l}\cdot{\alpha_2}^{\sum_{k=0}^{l-1}{(\dL'\cdot d')}^k}\\
&= {{\alpha_1}^{(\dL'\cdot d')}}^{l}\cdot{\alpha_2}^{\frac{(\dL'\cdot d')^l-1}{\dL'\cdot d'-1}}\\
&= {{\big(\alpha_1\cdot {\alpha_2}^{\frac{1}{\dL'\cdot d'-1}}\big)}^{(\dL'\cdot d')}}^{l}\cdot{\alpha_2}^{-\frac{1}{\dL'\cdot d'-1}}\\
&\leqslant \alpha^{{(\dL' \cdot d')}^l-1}.
\end{align*}
By Lemma~\ref{lemma:ADSbound1}, we conclude that
\[\Pi_{\lambda,d,\dL,\dR}(h,1^h) \leqslant \alpha^{\dL \cdot {(\dL'\cdot d')}^{h-1} - \dL }.\]

To analyze parameters for which $\Pi_{\lambda,d,\dL,\dR}(h,1^h)\rightarrow 0$, we need to compute $\alpha_1$ and $\alpha_2$ as functions of $p$, $d$, $\dL$ and $\dR$.
Note that
\begin{equation}
X_0 = \begin{cases} d'-2k & \mathrm{w.p.\ }\binom{\dR'}{k} p^k(1-p)^{\dR'-k},~~~ 0 \leqslant k < d',\\ -d'& \mathrm{w.p.\ }\sum_{k=d'}^{\dR'}\binom{\dR'}{k} p^k(1-p)^{\dR'-k}. \end{cases}
\end{equation}
Therefore,
\begin{align}
\alpha_1(p,d, \dL, \dR,t) = \sum_{k=0}^{d'-1}{{\binom{\dR'}{k}p^k(1-p)^{(\dR'-k)}} e^{-t(d'-2k)}}
+\bigg(\sum_{k=d'}^{\dR'}\binom{\dR'}{k}p^k(1-p)^{\dR'-k}\bigg)e^{td'},\label{eqn:c_1}
\end{align}
and
\begin{align}
  \alpha_2(p,d, \dL, \dR,t)= \binom{\dR'}{d'}\big((1-p)e^{-t}+pe^t\big)^{d'}. \label{eqn:c_2}
\end{align}
The above calculations give the following bound on $\Pi_{\lambda,d,\dL,\dR}(h,1^h)$.

\begin{lemma} \label{lemma:uniformBound}
Let $p\in (0,\frac{1}{2})$ and let $d,\dL,\dR\geqslant2$ s.t. $\dL'\cdot d'\geqslant2$. Denote by $\alpha_1$ and $\alpha_2$ the functions defined in (\ref{eqn:c_1})--(\ref{eqn:c_2}), and let
\[ \alpha=\min_{t\geqslant0}\bigg\{\alpha_1(p,d, \dL, \dR,t)\cdot \big(\alpha_2(p,d, \dL, \dR,t) \big)^{1/(\dL'\cdot d'-1)}\bigg\}.
\]
Then, for $h \in \mathds{N}$ and $\omega = {1}^h$, we have
\[\Pi_{\lambda,d,\dL,\dR}(h,\omega) \leqslant \alpha^{\dL \cdot {\dL'}^{h-1} - \dL }.\]
\end{lemma}
\noindent Note that if $\alpha<1$, then $\Pi_{\lambda,d,\dL,\dR}(h,1^h)$ decreases doubly exponentially as a function of $h$.

For $(2,16)$-regular graphs and $d\in\{3,4\}$, we obtain the following corollary.

\begin{corollary}
\label{cor:uniformBoundExample}
Let $\dL=2$, and $\dR=16$.
\begin{enumerate}
\item Let $d=3$ and $p\leqslant0.0067$. Then, there exists a constant $\alpha<1$ such that for every $h\in\N$ and $w=1^h$, \[\Pi_{\lambda,d,\dL,\dR}(h,1^h) \leqslant {\alpha^2}^{h-1}.\]
    \item Let $d=4$ and $p\leqslant0.0165$. Then, there exists a constant $\alpha<1$ such that for every $h\in\N$ and $w=1^h$, \[\Pi_{\lambda,d,\dL,\dR}(h,1^h) \leqslant {\alpha^3}^{h-1}.\]
\end{enumerate}
\end{corollary}

The bound on $p$ for which Corollary~\ref{cor:uniformBoundExample} applies grows with $d$. This fact confirms that analysis based on denser trees, i.e., $d$-trees with $d>2$ instead of skinny trees, implies better bounds on the error probability and higher lower bounds on the threshold. Also, for $d>2$, we may apply the analysis to $(2,\dR)$-regular codes; a case that is not applicable by the analysis of Arora \emph{et al.}~\cite{ADS09}.

\subsubsection{Improved Bounds Using Non-Uniform Weights}

The following lemma implies an improved bound for $\Pi_{\lambda,d,\dL,\dR}(h,\omega)$ using a non-uniform weight vector $\omega$.

\begin{lemma} \label{lemma:improvedBound}
Let $p\in (0,\frac{1}{2})$ and let $d,\dL,\dR\geqslant2$ s.t. $\dL'\cdot d'\geqslant2$. For $s \in \mathds{N}$ and a weight vector $\overline{\omega}\in\R_+^s$, let
\begin{equation} \label{eqn:MinE}
 \alpha = \min_{t \geqslant 0}\big\{\E e^{-tX_s}\big\} \cdot \bigg(\binom{\dR'}{d'}\big(2\sqrt{p(1-p)}\big)^{d'}\bigg)^{\frac{1}{\dL'\cdot d' -1}}.
\end{equation}
Let $\omega^{(\rho)} \in \R^h_+$ denote the concatenation of the vector $\overline{\omega}\in\R_+^s$ and the vector $(\rho,\ldots,\rho) \in \R_+^{h-s}$. Then, for every $h>s$ there exists a constant $\rho > 0$ such that
\begin{align*}
\Pi_{\lambda,d,\dL,\dR}(h,\omega^{(\rho)})\leqslant\bigg(\binom{\dR'}{d'}\big(2\sqrt{p(1-p)}\big)^{d'}\bigg)^{-\frac{\dL'}{\dL'\cdot d' -1}}\cdot \alpha^{\dL \cdot {(\dL'\cdot d')}^{h-s-1}}.
\end{align*}
\end{lemma}

\begin{proof}
See Appendix~\ref{app:proofofImprovedBound}.
\end{proof}

Consider a weight vector $\overline{\omega}$ with components $\overline{\omega}_l = \big((\dL-1)(d-1)\big)^l$. This weight vector has the effect that every level in a skinny tree $\tau$ contributes equally to $\mathrm{val}_{\overline{\omega}} (\tau;\lvert\lambda\rvert)$ (note that $\lvert\lambda\rvert\equiv \vec{1}$).
For $h>s$, consider a weight vector $\omega^{(\rho)} \in \mathds{R}_+^h$ defined by
\[\omega^{(\rho)}_l = \begin{cases} \overline{\omega}_l & \mathrm{if\ }0 \leqslant l <s,\\
\rho & \mathrm{if\ } s\leqslant l < h .\end{cases}\] Note that the first $s$ components of $\omega^{(\rho)}$ are geometric while the other components are uniform.

For a given $p$, $d$, $\dL$, and $\dR$, and for a concrete value $s$ we can compute the distribution of $X_s$ using the recursion in (\ref{eqn:Y_0})--(\ref{eqn:Y_l}). Moreover, we can also compute the value $\min_{t\geqslant0}\mathds{E}e^{-tX_s}$.
For $(2,16)$-regular graphs we obtain the following corollaries. Corollary~\ref{cor:improvedBoundExample3} is stated for the case where $d=3$, and Corollary~\ref{cor:improvedBoundExample4} is stated for the case where $d=4$.

\begin{corollary} \label{cor:improvedBoundExample3}
Let $p\leqslant p_0$, $d=3$, $\dL = 2$, and $\dR=16$.
For the following values of $p_0$ and $s$ shown in Table~\ref{table:thresholds3} it holds that there exists constants $\rho>0$ and $\alpha<1$ such that for every $h>s$,
\[\Pi_{\lambda,d,\dL,\dR}(h,\omega^{(\rho)}) \leqslant \frac{1}{420}\big(p(1-p)\big)^{-1}\cdot \alpha^{{2}^{h-s}}.\]
\begin{table}
\renewcommand{\arraystretch}{1.3}
\caption{Computed values of $p_0$ for finite $s$ in Corollary~\ref{cor:improvedBoundExample3} for a BSC.
Values are presented for $(\dL,\dR)=(2,16)$ and $d=3$.}\label{table:thresholds3}
\begin{center}
{\scriptsize
\begin{tabular}{||c|c||c||c|c||}
  \cline{1-2}
  \cline{4-5}
  $s$ & $p_0$   & & $s$ & $p_0$ \\
  \cline{1-2}
  \cline{4-5}
  \cline{1-2}
  \cline{4-5}
  $0$ & $0.0086$ & & $4$ & $0.0164$ \\
  \cline{1-2}
  \cline{4-5}
  $1$ & $0.011$ & & $5$ & $0.0171$ \\
  \cline{1-2}
  \cline{4-5}
  $2$ & $0.0139$  & & $6$ & $0.0177$\\
  \cline{1-2}
  \cline{4-5}
  $3$ & $0.0154$ & & $10$ & $0.0192$\\
  \cline{1-2}
  \cline{4-5}
  \cline{1-2}
  \cline{4-5}
\end{tabular}
} %\end scriptsize
\end{center}
\end{table}
\end{corollary}

\begin{corollary} \label{cor:improvedBoundExample4}
Let $p\leqslant p_0$, $d=4$, $\dL = 2$, and $\dR=16$.
For the following values of $p_0$ and $s$ shown in Table~\ref{table:thresholds4} it holds that there exists constants $\rho>0$ and $\alpha<1$ such that for every $h>s$,
\[\Pi_{\lambda,d,\dL,\dR}(h,\omega^{(\rho)}) \leqslant \frac{1}{60}\big(p(1-p)\big)^{-\frac{3}{4}}\cdot \alpha^{{3}^{h-s}}.\]
\begin{table}
\renewcommand{\arraystretch}{1.3}
\caption{Computed values of $p_0$ for finite $s$ in Corollary~\ref{cor:improvedBoundExample4} for a BSC. Values are presented for $(\dL,\dR)=(2,16)$ and $d=4$.}\label{table:thresholds4}
\begin{center}
{\scriptsize
\begin{tabular}{||c|c||c||c|c||}
  \cline{1-2}
  \cline{4-5}
  $s$ & $p_0$   & & $s$ & $p_0$ \\
  \cline{1-2}
  \cline{4-5}
  \cline{1-2}
  \cline{4-5}
  $0$ & $0.0218$ & & $4$ & $0.039$ \\
  \cline{1-2}
  \cline{4-5}
  $1$ & $0.0305$ & & $5$ & $0.0405$ \\
  \cline{1-2}
  \cline{4-5}
  $2$ & $0.0351$  & & $6$ & $0.0415$\\
  \cline{1-2}
  \cline{4-5}
  $3$ & $0.0375$ & & $10$ & $0.044$\\
  \cline{1-2}
  \cline{4-5}
  \cline{1-2}
  \cline{4-5}
\end{tabular}
} %\end scriptsize
\end{center}
\end{table}
\end{corollary}

Note that for a fixed $s$, the probability $\Pi_{\lambda,d,\dL,\dR}(h,\omega)$ decreases doubly exponentially as a function of $h$.

\subsection{Analysis for MBIOS Channels}\label{subsec:MBIOSbound}

Theorem~\ref{thm:main-bound-BSC} generalizes to MBIOS channels as follows.

\begin{theorem} \label{thm:main-bound-MBIOS}
Let $G$ denote a $(\dL,\dR)$-regular bipartite graph with girth $\Omega(\log N)$, and let $\calC(G)
  \subset \{0,1\}^N$ denote a Tanner code based on $G$ with minimum local distance $d^*$. Consider an MBIOS channel, and let $\lambda \in
  \R^N$ denote the LLR vector received from the channel given $c=0^N$.  Let
  $\gamma \in \R$ denote a random variable independent and identically distributed to components of $\lambda$. Then, for any $(\dL,\dR)$ and $2\leqslant d \leqslant d^*$ s.t. $(\dL-1)(d-1)\geqslant 2$,  LP decoding succeeds with probability at least $1-\exp(-N^\delta)$ for some
  constant $0<\delta<1$, provided that
\begin{equation*}
  \min_{t\geqslant0}\bigg\{ \E e^{-tX_0}\cdot\bigg( \binom{\dR-1}{d-1}\big(\E e^{-t\gamma}\big)^{(d-1)} \bigg)^{\frac{1}{(\dL-1)(d-1)-1}} \bigg\}<1.
\end{equation*}
where $X_0=\sum_{i=1}^{d-1}\min^{[i]}\{\gamma^{(1)},\ldots,\gamma^{(\dR-1)}\}$ and the random variables $\gamma^{(i)}$ are independent and distributed identically to $\gamma$.
\end{theorem}

%%%%%%%%%%%%%%%%%%%%%%%%%%%%%%%%%%%%%%%%%%%%%%%%%%%%%%%%%%%%%%%%%%%%%%%%%%%%%%%%%%
\section{Conclusions and Discussion} \label{sec:conclusion}
%%%%%%%%%%%%%%%%%%%%%%%%%%%%%%%%%%%%%%%%%%%%%%%%%%%%%%%%%%%%%%%%%%%%%%%%%%%%%%%%%%

We have presented a new combinatorial characterization of local optimality for irregular Tanner codes w.r.t. any MBIOS channel. This characterization provides an ML certificate and an LP certificate for a given codeword. Moreover, the certificate can be efficiently computed by a dynamic programming algorithm. Two applications of local optimality
are presented based on this new characterization.
(i)~A new message-passing decoding algorithm for irregular LDPC codes, called \NWMS. The \NWMS\ decoding algorithm is guaranteed to find the locally-optimal codeword if one exists.
(ii)~Bounds for LP decoding failure to decode the transmitted codeword are proved in the case of regular Tanner codes.
We discuss these two applications of local optimality in the following subsections.

\subsection{Applying \NWMS\ Decoding Algorithm to Regular LDPC Codes}\label{subsec:DiscussionNWMS}
The \NWMS\ decoding algorithm is a generalization of the min-sum decoding algorithm (a.k.a. max-product algorithm in the probability-domain) and other BP-based decoding algorithms in the following sense. When restricted to regular Tanner graphs and exponential level weights (to cancel the normalization in the variable node degrees), the \NWMS\ decoding algorithm reduces to the standard min-sum decoding algorithm~\cite{WLK95,Wib96}. Reductions of the \NWMS\ decoding algorithm to other BP-based decoding  algorithms (see, e.g., attenuated max-product \cite{FK00} and normalized BP-based \cite{CF02,CDEFH05}) can be obtained by other weight level functions.

Many works on the BP-based decoding algorithms study the convergence of message passing algorithms to an optimum solution on various settings (e.g.,~\cite{WF01,WJW05,RU01,JP11}). However, bounds on the running time required to decode have not been proven for these algorithms. The analyses of convergence in these works often rely on the existence of a single optimal solution in addition to other conditions such as: a single cycle, large girth, large reweighing coefficient, consistency conditions, etc. On the other hand, the \NWMS\ decoding algorithm is guaranteed to compute the ML codeword within $h$ iterations if a locally optimal certificate with height parameter $h$ exists for some codeword. Moreover, the certificate can be computed efficiently (see Algorithm~\ref{alg:verify}).

In previous works~\cite{ADS09,HE11}, the probability that a locally optimal certificate with height parameter $h$ exists for some codeword was investigated for regular LDPC codes with $h<\frac{1}{4}\girth(G)$.
Consider a $(\dL,\dR)$-regular LDPC code whose Tanner graph $G$ has logarithmic girth, let $h<\frac{1}{4}\girth(G)$ and define a constant weight vector $w\triangleq1^h$. In that case, the message normalization by variable node degrees has the effect that each level of variable nodes in a $2$-tree contributes equally to the cost of the $w$-weighted value of the $2$-tree. Hence, the set $\calB_2^{(w)}$ of \PNW\ deviations is equal to the set of $(\dL-1)$-exponentially weighted skinny trees~\cite{ADS09,HE11}. Following Equation~(\ref{eqn:MPfailureBound}), we conclude that the previous bounds on the probability that a locally optimal certificate exists~\cite{ADS09,HE11} apply also to the probability that the \NWMS\ and \certNWMS\ decoding algorithms successfully decode the transmitted codeword.

Consider $(3,6)$-regular LDPC codes whose Tanner graphs $G$ have logarithmic girth, and let $h=\frac{1}{4}\girth(G)$ and $w=1^h$. Then, $\NWMS(\lambda,h,w)$ and $\certNWMS(\lambda,h,w)$ succeed in recovering the transmitted codeword with probability at least $1-\exp(-N^\delta)$ for some constant $0<\delta<1$ in the following cases:
\begin{compactenum}[(1)]
\item In a BSC with crossover probability $p<0.05$ (implied by~\cite[Theorem~5]{ADS09}).
\item In a BI-AWGN channel with $\frac{E_b}{N_0}\geqslant 2.67$dB (implied by~\cite[Theorem~1]{HE11}).
\end{compactenum}

It remains to explore good weighting schemes (choice of vectors $w$) for specific families of irregular LDPC codes, and prove that a locally optimal codeword exists with high probability provided that the noise is bounded.
Such a result would imply that the \NWMS\ decoding algorithm is a good, efficient replacement for LP decoding.

\subsection{Bounds on the Word Error Probability for LP Decoding of Tanner Codes}
\label{subsec:DiscussionBounds}
In Section~\ref{sec:bounds} we proved bounds on the word error probability of LP decoding of regular Tanner codes. In particular, we considered a concrete example of $(2,16)$-regular Tanner codes with $[16,11,4]$-Hamming codes as local codes and Tanner graphs with logarithmic girth. The rate of such codes is at least $0.375$.
For the case of a BSC with crossover probability $p$, we prove a lower bound of $p^*=0.044$ on the noise threshold. Below that threshold the word error probability decreases doubly exponential in the girth of the Tanner graph.

Most of the research on the error correction of Tanner codes deals with families of expander Tanner codes.
How do the bounds presented in Section~\ref{sec:bounds} compare with results on expander Tanner codes?
The error correction capability of expander codes depends on the expansion, thus a fairly large degree and huge block lengths are required to achieve good error correction. Our example for which results are stated in Theorem~\ref{thm:main-bound-BSC}(1) and ~\ref{thm:main-bound-BSC}(2) relies only on a $16$-regular graph with logarithmic girth.
Sipser and Spielman~\cite{SS96} studied Tanner codes based on expander graphs and analyzed a simple bit-flipping iterative decoding algorithm. Their novel scheme was later improved, and it was shown that expander Tanner codes can even \emph{asymptotically} achieve capacity in a BSC with an iterative decoding bit-flipping scheme~\cite{Zem01,BZ02,BZ04}. In these works, a worst-case analysis (for an adversarial channel) was performed as well.

The best result for iterative decoding of such expander codes, reported by Skachek and Roth~\cite{SR03}, implies a lower bound of $p^* = 0.0016$ on the threshold of a certain iterative decoder for rate $0.375$ codes.
Feldman and Stein~\cite{FS05} proved that LP decoding can \emph{asymptotically} achieve capacity with a special family of expander Tanner codes. They also presented a worst-case analysis, which in the case of a code rate of $0.375$, proves that LP decoding can recover any pattern of at most $0.0008N$ bit flips. This implies a lower bound of $p^* = 0.0008$ on the noise threshold.
These analyses yield overly pessimistic predictions for the average case (i.e., a BSC). Theorem~\ref{thm:main-bound-BSC}(2) deals with average case analysis and implies that LP decoding can correct up to $0.044N$ bit flips with high probability.
Furthermore, previous iterative decoding algorithms for expander Tanner codes deal only with bit-flipping channels. Our analysis for LP decoding applies to any MBIOS channel, in particular, it can be applied to the BI-AWGN channel.

However, the lower bounds on the noise threshold proved for Tanner codes do not improve the best previous bounds for regular LDPC codes with the same rate. An open question is whether using deviations denser than skinny trees for Tanner codes can beat the best previous bounds for regular LDPC codes~\cite{ADS09,HE11}. In particular, for a concrete family of Tanner codes with rate $\frac{1}{2}$, it would be interesting to prove lower bounds on the threshold of LP decoding that are larger than $p^*=0.05$ in the case of a BSC, and $\sigma^*=0.735$ in the case of a BI-AWGN channel (upper bound smaller than $\frac{E_b}{N_0}=2.67$dB).

%%%%%%%%%%%%%%%%%%%%%%%%%%%%%%%%%%%%%%%%%%%%%%%%%%%%%%%%%%%%%%%%%%%%%%%%
\section*{Acknowledgement}
%%%%%%%%%%%%%%%%%%%%%%%%%%%%%%%%%%%%%%%%%%%%%%%%%%%%%%%%%%%%%%%%%%%%%%%%
The authors would like to thank Pascal O. Vontobel for many valuable comments and suggestions that improved the paper. We appreciate the helpful comments that were made by the reviewers.

%\bibliographystyle{plain}
%\small
%\clearpage
%\bibliographystyle{abbrv}
\bibliographystyle{alpha}
%\bibliography{ECCbib}

\newcommand{\etalchar}[1]{$^{#1}$}

\appendix

%%%%%%%%%%%%%%%%%%%%%%%%%%%%%%%%%%%%%%%%%%%%%%%%%%%%%%%%%%%%%%%
\section{Constructing Codewords from Projection of Weighted\\ Trees -- Proof of Lemma~\ref{lemma:IntegralDecomposition}} \label{sec:decomposition}
%%%%%%%%%%%%%%%%%%%%%%%%%%%%%%%%%%%%%%%%%%%%%%%%%%%%%%%%%%%%%%%
In this section we prove Lemma~\ref{lemma:IntegralDecomposition}, the
key structural lemma in the proof of Theorem~\ref{thm:MLsufficient}.
This lemma states that every codeword of a Tanner code is a finite
sum of projections of weighted trees in the computation trees of $G$.

Throughout this section, let $\mathcal{C}(G)$ denote a Tanner code
with minimum local distance $d^*$, let $x$ denote a nonzero codeword,
let $h$ denote some positive integer, and let $w\in \R_+^h\setminus
\{0^h\}$ denote level weights.

The proof of Lemma~\ref{lemma:IntegralDecomposition} is based on
Lemmas~\ref{lemma:prefixDecomposition}--\ref{lemma:skiniesDecomposition}
  and Corollary~\ref{cor:projEQexp}.  Lemma~\ref{lemma:prefixDecomposition} states
  that every codeword $x\in\mathcal{C}(G)$ can be decomposed into a
  set of weighted path-prefix trees. The number of trees in the
  decomposition equals $\lVert x\rVert_1$.
  Lemma~\ref{lemma:skiniesDecomposition} states that every weighted
  path-prefix tree is a convex combination of weighted $d$-trees. This
  lemma implies that the projection of a weighted path-prefix tree is
  equal to the expectation of projections of weighted $d$-trees.

  For a codeword $x\in \mathcal{C}(G) \subset \{0,1\}^N$, let $\calV_x
  \triangleq \{v\in\calV~\vert~x_v=1\}$. Let $G_x$ denote the subgraph of
  the Tanner graph $G$ induced by $\calV_x\cup\mathcal{N}(\calV_x)$. Note that the degree of every local-code node in $G_x$ is at least $d$.

\begin{lemma}\label{lemma:prefixDecomposition}
  For every codeword $x \neq 0^N$, for every weight vector $w \in
  \R_+^h$, and for every variable node $v\in\calV$, it holds that
\begin{equation*}
  x_v = \sum_{r\in\calV_x}\pi_{G,\calT_r^{2h}(G_x),w}(v).
\end{equation*}
\end{lemma}

\begin{proof}
If $x_v=0$, then $\pi_{G,\calT_r^{2h}(G_x),w}(v)=0$. It remains to show that equality holds for variable nodes $v\in\calV_x$.

Consider an all-one weight vector $\eta=1^h$. Construct a path-suffix
tree rooted at $v$. The set of nodes of a path-suffix tree rooted at $v$ contains paths that end at node $v$ (in contrast to path-prefix trees where the set of nodes contains paths that start at the root).
Level $\ell$ of this path-suffix tree consists of
all backtrackless paths in $G_x$ of length $\ell$ that end at node $v$
(see Figure~\ref{fig:decompositionProof}).  We denote this level by
$P_\ell(v)$.

We use the same notational convention for $\eta$ as for $w$ in Definition~\ref{def:weightedSubtree}, i.e., $\eta_\calT$ denotes a weight function based on weight vector $\eta$ for variable paths in $\calT$.
We claim that for every $v\in\calV_x$ and $1 \leqslant \ell \leqslant 2h$,
\begin{equation} \label{eqn:unitLevel}
\sum_{p\in P_\ell(v)}\eta_{\calT_{s(p)}^{2h}}(p)=\frac{1}{h}.
\end{equation}
The proof is by induction on $\ell$. The induction basis, for $\ell=1$, holds because $\lvert P_1(v)\rvert=\deg_G(v)$ and $\eta_{\calT_{s(p)}^{2h}}(p)=\frac{1}{\lVert \eta\rVert_1}\cdot\frac{1}{\deg_G(v)}=\frac{1}{h}\cdot\frac{1}{\deg_G(v)}$ for every $p\in P_1(v)$.
\begin{figure}
  \begin{center}
 \includegraphics[width=0.7\textwidth]{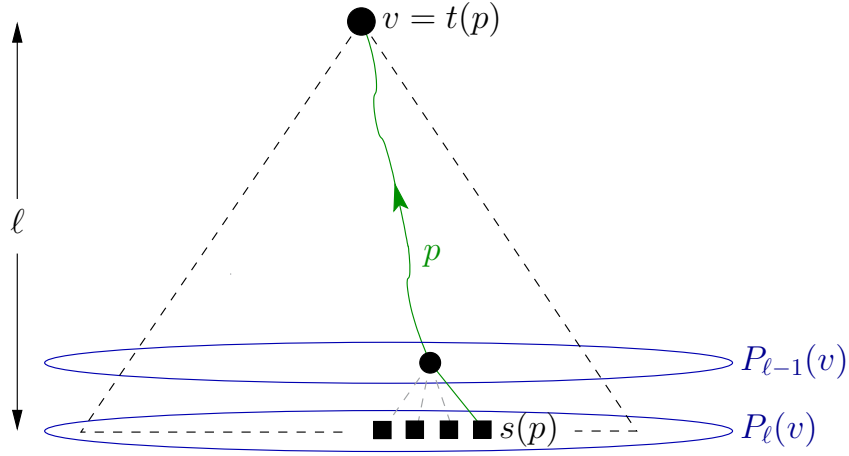}
 \caption{Set of all backtrackless paths $P_{\ell}(v)$ as augmentation
   of the set $P_{\ell-1}(v)$ as viewed by the path-suffix tree of
   height $\ell$ rooted at $v$, in the proof of
   Lemma~\ref{lemma:IntegralDecomposition}. Note that if $\ell$ is odd, then every path that ends at variable node $v$ starts at a local-code node. If $\ell$ is even, then every path that ends at variable node $v$ starts at a variable node.}
  \label{fig:decompositionProof}
  \end{center}
\end{figure}
The induction step is proven as follows. For each $p\in P_\ell(v)$, let $\aug(p)\triangleq\big\{q\in P_{\ell+1}(v)\ \big\vert\ p\mathrm{\ is\ a\ suffix\ of\ }q\big\}$. Note that $\lvert\aug(p)\rvert =\deg_{G_x}\big(s(p)\big)-1$. Moreover, for each $q\in\aug(p)$,
\begin{equation}
\frac{\eta_{\calT_{s(q)}^{2h}}(q)}{\eta_{\calT_{s(p)}^{2h}}(p)}=\frac{1}{\deg_{G_x}\big(s(p)\big)-1}.
\end{equation}
Hence, \[\sum_{q\in\aug(p)}\eta_{\calT_{s(q)}^{2h}}(q) = \eta_{\calT_{s(p)}^{2h}}(p).\] Finally, $P_{\ell+1}(v)$ is the disjoint union of $\bigcup_{p\in P_\ell(v)}\aug(p)$. It follows that
\begin{equation}
\sum_{q\in P_{\ell+1}(v)}\eta_{\calT_{s(q)}^{2h}}(q) = \sum_{p\in P_{\ell}(v)}\eta_{\calT_{s(p)}^{2h}}(p).
\end{equation}
By the induction hypothesis we conclude that $\sum_{q\in
  P_{\ell+1}(v)}\eta_{\calT_{s(q)}^{2h}}(q)=1/h$, as required.  Note
that the sum of weights induced by $\eta$ on each level is $1/h$, both
for levels of paths beginning in variable nodes and in local-code
nodes. In the rest of the proof we focus only on even levels that
start at variable nodes. We now claim that
\begin{equation}\label{eqn:w_ell}
\sum_{p\in P_{2\ell}(v)}w_{\calT_{s(p)}^{2h}}(p)=\frac{w_\ell}{\wone}.
\end{equation}
Indeed, by Definition~\ref{def:weightedSubtree} it holds that $w_{\calT_{s(p)}^{2h}}(p)=\eta_{\calT_{s(p)}^{2h}}(p)\cdot \frac{w_\ell}{\wone}\cdot h$  for every $p\in P_{2\ell}(v)$.
Therefore, Equation~(\ref{eqn:w_ell}) follows from Equation~(\ref{eqn:unitLevel}).

The lemma follows because for every $v\in\calV_x$,
\begin{eqnarray*}
\sum_{r\in\calV_x}\pi_{G,\calT_{r}^{2h}(G_x),w}(v) &=& \sum_{\ell=1}^{h}\sum_{p\in P_{2\ell}(v)}w_{\calT_{s(p)}^{2h}}(p)\\
&=&\sum_{\ell=1}^{h}\frac{w_\ell}{\wone} = 1.
\end{eqnarray*}
\end{proof}

\begin{lemma} \label{lemma:skiniesDecomposition} Consider a subgraph
  $G_x$ of a Tanner graph $G$, where $x\in \calC(G)\setminus \{0^N\}$.
  Then, for every variable node $r \in G_x$, every positive integer
  $h$, every $2 \leqslant d \leqslant d^*$, and every weight vector $w
  \in \R_+^h$, it holds that
  \[ w_{\calT_r^{2h}(G_x)} = \E_{\rho_r}\big[
w_\calT\big]\]
where $\rho_r$ is the  uniform distribution over $\calT[r,2h,d](G_x)$.
\end{lemma}

\begin{proof}
  Let $G_x=(\calV_x\cup\calJ_x,E_x)$ and let $w_{\calT_r^{2h}(G_x)}$
  denote a $w$-weighted path-prefix tree rooted at node $r$ with
  height $2h$.  We claim that the expectation of $w$-weighted $d$-trees
  $w_\calT\in\calT[r,2h,d](G_x)$ equals $w_{\calT_r^{2h}(G_x)}$ if
  $w_\calT$ is chosen uniformly at random.

  Let $\rho_r$ denote the uniform distribution over $\calT[r,2h,d](G_x)$.
  A random $d$-tree in $\calT[r,2h,d](G_x)$ can be sampled according
  to $\rho_r$ as follows.  Start from the root $r$. For each variable
  path, take all its augmentations\footnote{Note the difference between an augmentation of a
    variable path in a path-prefix tree and a path-suffix tree.
In a path-prefix tree, an augmentation appends a node to the end of the path.
In a path-suffix tree, an augmentation adds a node before the beginning of the path.}, and for each local-code path choose $d-1$ distinct
  augmentations uniformly at random. Let $\calT\in\calT[r,2h,d](G_x)$
  denote such a random $d$-tree, and consider a variable path $p \in
  \calT_r^{2h}(G_x)$.  Then,
\begin{equation}\label{eqn:ew1}
\Pr_{\rho_r}\{p\in\calT\} = \prod_{\{q\in\PPrefix(p)~\vert~t(q)\in\calJ_x\}}
    \frac{d-1}{\deg_{G_x}(t(q))-1}.
\end{equation}

Note the following two observations: (i)~if $p\notin\calT$, then $w_\calT(p)=0$, and (ii)~if $p\in\calT$, then the value of $w_\calT(p)$ is constant, i.e., $w_\calT(p)=w_{\calT'}(p)$ for all $\calT'$ such that $p\in\calT'$.
Let $\alpha(p)$ denote this constant, i.e., $\alpha(p)\triangleq w_\calT(p)$ for some $\calT\in\calT[r,2h,d](G_x)$ such that $p\in\calT$. From the two observations above we have
\begin{equation}\label{eqn:ew3}
\E_{\rho_r}\big[w_\calT(p)\big] =
\alpha(p) \cdot \Pr_{\rho_r}\{p\in\calT\}.
\end{equation}
Note that for a variable path $p\in\calT$, $\lvert p\rvert$ is even because $\calT$ is rooted at a variable node $r$.
By Definition~\ref{def:weightedSubtree}, for a variable path $p\in\calT$ we have
\begin{align}
\alpha(p)=w_{\calT}(p)=\frac{w_{\lvert p\rvert/2}}{\wone}\cdot
\frac{1}{\deg_{G_x}(t(p))}\cdot \frac{1}{(d-1)^{\lvert p\rvert /2}} \cdot \prod _{\{q\in\PPrefix(p)~\vert~t(q)\in\calV_x\}} \frac{1}{\deg_{G_x}(t(q))-1}.\label{eqn:ew2}
\end{align}
By substituting~(\ref{eqn:ew1}) and~(\ref{eqn:ew2}) in~(\ref{eqn:ew3}), we conclude that
\begin{align*}
\E_{\rho_r}\big[w_\calT(p)\big]&=
\frac{w_{\lvert p\rvert/2}}{\wone}\cdot\frac{1}{\deg_{G_x}(t(p))}\cdot\prod_{q\in\PPrefix(p)}
    \frac{1}{\deg_{G_x}(t(q))-1}\\
&=  w_{\calT_r^{2h}(G_x)}(p).
\end{align*}
\end{proof}

\begin{corollary}\label{cor:projEQexp}
For every positive integer $h$, every $2 \leqslant d \leqslant d^*$, and every weight vector $w \in \R_+^h$, it holds that
\begin{equation*}
\pi_{G,\calT_r^{2h}(G_x),w} = \E_{\rho_r}\big[\pi_{G,\calT,w}\big]
\end{equation*}
where $\rho_r$ is the  uniform distribution over $\calT[r,2h,d](G_x)$.
\end{corollary}
\begin{proof}
By definition of $\pi_{G,\calT_r^{2h}(G_x),w}$, we have
\begin{equation}\label{eqn:projEqexp1}
\pi_{G,\calT_r^{2h}(G_x),w}(v) = \sum_{\{p\in\calT_r^{2h}(G_x)~\vert~t(p)=v\}}w_{\calT_r^{2h}(G_x)}(p).
\end{equation}
By Lemma~\ref{lemma:skiniesDecomposition} and linearity of expectation we have
\begin{align}\label{eqn:projEqexp2}
\sum_{\{p\in\calT_r^{2h}(G_x)~\vert~t(p)=v\}}w_{\calT_r^{2h}(G_x)}(p)&=\sum_{\{p\in\calT_r^{2h}(G_x)~\vert~t(p)=v\}}\E_{\rho_r}\big[w_\calT(p)\big]\nonumber\\
&=\E_{\rho_r}\bigg[\sum_{\{p\in\calT_r^{2h}(G_x)~\vert~t(p)=v\}}w_\calT(p)\bigg].
\end{align}\label{eqn:projEqexp3}
Now, for variable paths $p$ that are not in a $d$-tree $\calT$, $w_\calT(p)=0$. Hence, if a $d$-tree $\calT$ is a subtree of $\calT_r^{2h}(G_x)$, then
\begin{eqnarray}
\sum_{\{p\in\calT_r^{2h}(G_x)~\vert~t(p)=v\}}w_\calT(p) &=& \sum_{\{p\in\calT~\vert~t(p)=v\}}w_\calT(p)\nonumber\\
&=& \pi_{G,\calT,w}(v). \label{eq:18}
\end{eqnarray}
From Equations~(\ref{eqn:projEqexp1})--(\ref{eq:18}) we conclude that for every $v\in\calV$,
\begin{equation*}
\pi_{G,\calT_r^{2h}(G_x),w}(v) = \E_{\rho_r}\big[\pi_{G,\calT,w}(v)\big].
\end{equation*}
\end{proof}
Before proving Lemma~\ref{lemma:IntegralDecomposition}, we state a
lemma from probability theory.
\begin{lemma}\label{lemma:sumOfExp}
Let $\{\rho_r\}_{i=1}^K$ denote $K$ probability distributions.
Let $\rho \eqdf \frac{1}{K} \sum_{r=1}^K \rho_r$.
Then,
\begin{equation*}
\sum_{r=1}^K\E_{\rho_r}[x]=K\cdot\E_{\rho}[x].
\end{equation*}
\end{lemma}

\begin{proof}[Proof of Lemma~\ref{lemma:IntegralDecomposition}]
By Lemma~\ref{lemma:prefixDecomposition} and Corollary~\ref{cor:projEQexp} we have for every $v\in\calV_x$
\begin{align}\label{eqn:IntegralDecomposition1}
x_v &= \sum_{r\in\calV_x}\pi_{G,\calT_r^{2h}(G_x),w}(v)\\
&=\sum_{r\in\calV_x}\E_{\rho_r}\big[\pi_{G,\calT,w}\big].\nonumber
\end{align}
Let $\rho$ denote the distribution defined by
$\rho\eqdf \frac{1}{\xone} \cdot \sum_{r\in \calV_x} \rho_r$.
By Lemma~\ref{lemma:sumOfExp} and Equation~(\ref{eqn:IntegralDecomposition1}),
\[
x_v = \xone \cdot \E_{\rho}\big[\pi_{G,\calT,w}\big],
\]
and the lemma follows.
\end{proof}

%%%%%%%%%%%%%%%%%%%%%%%%%%%%%%%%%%%%%%%%%%%%%%%%%%%%%%%%%%%%%%%%%%%%%%%
\section{Symmetry of \NWMS~--~Proof of Lemma~\ref{lemma:NWMSsymmetry}} \label{app:NWMSSymmetryProof}
%%%%%%%%%%%%%%%%%%%%%%%%%%%%%%%%%%%%%%%%%%%%%%%%%%%%%%%%%%%%%%%%%%%%%%%%

\begin{proof}
Let $\mu_{v\rightarrow C}^{(l)}[\lambda]$ denote the message sent from $v$ to $C$ in iteration $l$ given an input $\lambda$. Let $\mu_{C\rightarrow v}^{(l)}[\lambda]$ denote the corresponding message from $C$ to $v$. From the decision of \NWMS\ in Line 12, it's sufficient to prove that $\mu_{v\rightarrow C}^{(l)}[\lambda]=(-1)^{x_v}\cdot\mu_{v\rightarrow C}^{(l)}[(-1)^x\ast\lambda]$ and $\mu_{C\rightarrow v}^{(l)}[\lambda]=(-1)^{x_v}\cdot\mu_{C\rightarrow v}^{(l)}[(-1)^x\ast\lambda]$ for every $0\leqslant l\leqslant h-1$.

The proof is by induction on $l$. The induction basis, for $l=-1$, holds because $\mu_{C\rightarrow v}^{(-1)}[\lambda]=(-1)^{x_v}\cdot\mu_{C\rightarrow v}^{(-1)}[(-1)^x\ast\lambda]=0$ for every codeword $x$.

The induction step is proven as follows. By induction hypothesis we have
\begin{eqnarray*}
\mu_{v\rightarrow C}^{(l)}[\lambda]&=&\frac{w_{h-l}}{\deg_G(v)}\lambda_v+\frac{1}{\deg_G(v)-1}
\sum_{C'\in\calN(v)\setminus\{C\}}\mu_{C'\rightarrow v}^{(l-1)}[\lambda]\\
&=&(-1)^{x_v}\cdot\bigg(\frac{w_{h-l}}{\deg_G(v)}(-1)^{x_v}\lambda_v+\frac{1}{\deg_G(v)-1}
\sum_{C'\in\calN(v)\setminus\{C\}}\mu_{C'\rightarrow v}^{(l-1)}[(-1)^x\ast\lambda]\bigg)\\
&=&(-1)^{x_v}\cdot\mu_{v\rightarrow C}^{(l)}[(-1)^x\ast\lambda].
\end{eqnarray*}
For check to variable messages we have by the induction hypothesis,
\begin{align*}
\mu_{C\rightarrow v}^{(l)}[\lambda]=&
\bigg(\prod_{u\in\calN(C)\setminus\{v\}}\textrm{sign}\big(\mu_{u\rightarrow C}^{(l)}[\lambda]\big)\bigg)\cdot\min_{u\in\calN(C)\setminus\{v\}}\big\{\big\lvert\mu_{u\rightarrow C}^{(l)}[\lambda]\big\rvert\big\}\\
=&\bigg(\prod_{u\in\calN(C)\setminus\{v\}}\textrm{sign}\big((-1)^{x_u}\cdot\mu_{u\rightarrow C}^{(l)}[(-1)^x\ast\lambda]\big)\bigg)\\
&\cdot\min_{u\in\calN(C)\setminus\{v\}}\big\{\big\lvert(-1)^{x_u}\cdot\mu_{u\rightarrow C}^{(l)}[(-1)^x\ast\lambda]\big\rvert\big\}\\
=&\bigg(\prod_{u\in\calN(C)\setminus\{v\}}(-1)^{x_u}\bigg)\cdot\mu_{C\rightarrow v}^{(l)}[(-1)^x\ast\lambda].
\end{align*}
Because $x$ is codeword, for every single parity-check $C$ we have $\prod_{u\in\calN(C)\setminus\{v\}}(-1)^{x_u}=(-1)^{x_v}$. Therefore, $\mu_{C\rightarrow v}^{(l)}[\lambda]=(-1)^{x_v}\cdot\mu_{C\rightarrow v}^{(l)}[(-1)^x\ast\lambda]$ and the claim follows.
\end{proof}

%%%%%%%%%%%%%%%%%%%%%%%%%%%%%%%%%%%%%%%%%%%%%%%%%%%%%%%%%%%%%%%%%%%%%%%%
\section{Optimal Valid Subconfigurations in the Execution\\ of
\NWMS2} \label{app:NWMS2_DP}
%%%%%%%%%%%%%%%%%%%%%%%%%%%%%%%%%%%%%%%%%%%%%%%%%%%%%%%%%%%%%%%%%%%%%%%%

The description of algorithm \NWMS2 as a dynamic programming algorithm deals with the computation of optimal valid configurations and subconfigurations. In this appendix we define optimal valid subconfigurations and prove invariants for the messages of algorithm \NWMS2.

Denote by $\calT_{C\rightarrow v}^{2l+2}$ a path prefix tree of $G$ rooted at node $v$ with height $2l+2$ such that all paths must start with edge $(v,C)$ (see Figure~\ref{fig:substructure_C-v}). Denote by $\calT_{v\rightarrow C}^{2l+1}$ a path prefix tree of $G$ rooted at node $C$ with height $2l+1$ such that all paths  start with edge $(C,v)$ (see Figure~\ref{fig:substructure_v-C}).
\begin{figure}
\centering
\subfigure[Substructure $\calT_{C\rightarrow v}^{2l+2}$ for $0\leqslant l\leqslant h-1$.]{\includegraphics[width=0.48\textwidth]{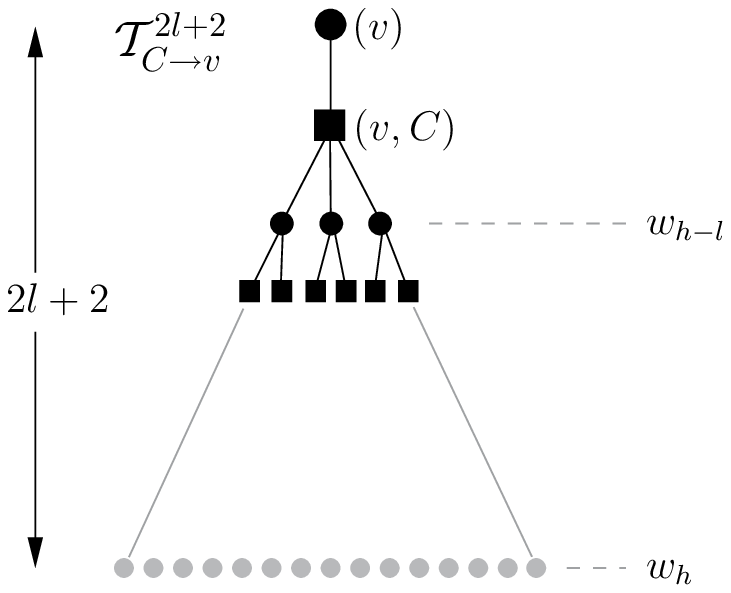}
\label{fig:substructure_C-v}}
\subfigure[Substructure $\calT_{v\rightarrow C}^{2l+1}$ for $0\leqslant l\leqslant h-1$.]{\includegraphics[width=0.48\textwidth]{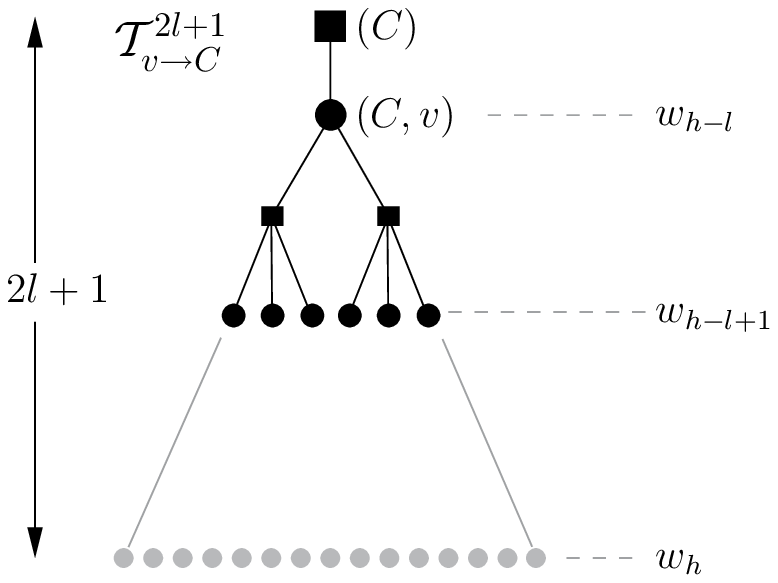}
\label{fig:substructure_v-C}}
\label{fig:substructures}
\caption{Substructures of a path-prefix tree $\calT_r^{2h}(G)$ in a dynamic programming that computes optimal configurations in $\calT_r^{2h}(G)$.}
\end{figure}

Consider the message $\mu_{C\rightarrow v}^{(2l+2)}$. It is determined by the messages sent along the edges of $\calT_v^{2l+2}(G)$ that hang from the edge $(v,C)$. We introduce the following notation of this subtree (see Figure~\ref{fig:substructureReasoning}).
Consider a path-prefix tree $\calT_r^{2h}(G)$ and a variable path $p$ such that
\begin{compactenum}[(i)]
%\begin{enumerate}[(i)]
\item $p$ is a path from root $r$ to a variable node $v$,
\item the last edge in $p$ is $(C',v)$ for $C'\neq C$, and
\item the length of $p$ is $2(h-l-1)$.
%\end{enumerate}
\end{compactenum}
In such a case, $\calT_{C\rightarrow v}^{2l+2}$ is isomorphic to the subtree of $\calT_r^{2h}$ hanging from $p$ along the edge $\big(p,p\circ(v,C)\big)$. Hence, we say that $\calT_{C\rightarrow v}^{2l+2}$ is a \emph{substructure} of $\calT_r^{2h}(G)$.
Similarly, if there exists a backtrackless path $q$ in $G$ from $r$ to $C$ with length $2(h-l)-1$ that does not end with edge $(v,C)$, we say that $\calT_{v\rightarrow C}^{2l+1}$ is a \emph{substructure} of $\calT_r^{2h}(G)$.
\begin{figure}
  \begin{center}
 \includegraphics[width=0.5\textwidth]{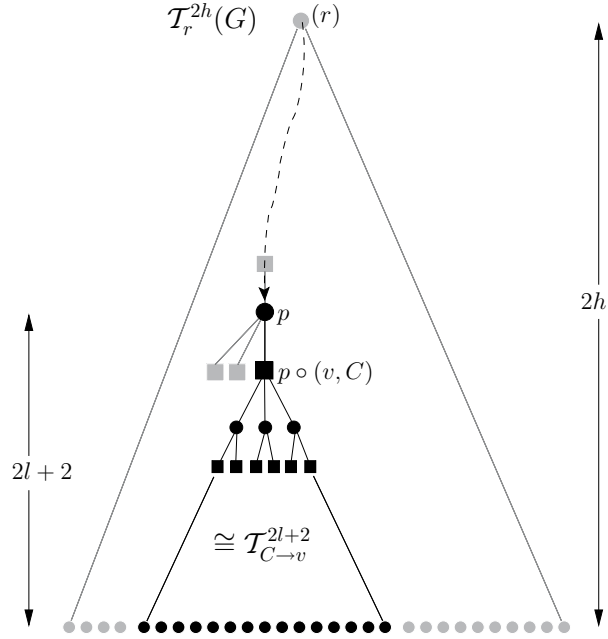}
 \caption{$\calT_{C\rightarrow v}^{2l+2}$ as a substructure isomorphic to a subtree of the path-prefix tree $\calT_r^{2h}$.}
  \label{fig:substructureReasoning}
  \end{center}
\end{figure}

Let $\calT_\sub$ denote a substructure $\calT_{C\rightarrow v}^{2l+2}$ or $\calT_{v\rightarrow C}^{2l+1}$. A binary assignment $z \in \{0,1\}^{\vert\hat\calV(\calT_\sub)\rvert}$ to variable paths $\hat{\calV}(\calT_{\sub})$ is a \emph{valid subconfiguration} if it satisfies every parity-check path $q\in\calT_\sub$ with $\lvert q\rvert\geqslant1$. We denote the set of valid subconfigurations of $\calT_\sub$ by
$\config(\calT_\sub)$.

Define the weight of a variable path $q\in\hat\calV(\calT_\sub)$ w.r.t. level weights $w=(w_1,\ldots,w_h)\in\R_+^h$ by
\begin{align*}
\calW_{\sub}(\calT_\sub,q)\triangleq \frac{w_{h-l-1+\lceil\lvert q\rvert/2\rceil}}{\deg_G\big(t(q)\big)}\cdot \prod_{q'\in\PPrefix(q)\cap\hat{\calV}(\calT_\sub)}\frac{1}{\deg_G\big(t(q')\big)-1}.
\end{align*}

The \emph{weight of a valid subconfiguration $z$}  for a substructure $\calT_\sub$
is defined by
\[
\calW_{\sub}(\calT_\sub,z) \triangleq \sum_{\{q\in\hat\calV(\calT_\sub)~\vert~\lvert q\rvert\geqslant1\}} \lambda_{t(q)}(z_q)\cdot \calW_{\sub}(\calT_\sub,q).
\]

Define the \emph{minimum weight of substructures $\calT_{v\rightarrow C}^{2l+1}$ and $\calT_{C\rightarrow v}^{2l+2}$ for $a\in\{0,1\}$} as follows.
\begin{align*}
\calW_{\sub}^{\min}(\calT_{v\rightarrow C}^{2l+1},a) \triangleq \min\bigg\{\calW_{\sub}(\calT_{v\rightarrow C}^{2l+1},z)~\bigg\vert~\begin{aligned}&z\in\config(\calT_{v\rightarrow C}^{2l+1}),\\ &z_{(C,v)}=a\end{aligned}\bigg\},
\end{align*}
and
\begin{align*}
\calW_{\sub}^{\min}(\calT_{C\rightarrow v}^{2l+2},a) \triangleq \min\bigg\{\calW_{\sub}(\calT_{C\rightarrow v}^{2l+2},z)~\bigg\vert~\begin{aligned}&z\in\config(\calT_{C\rightarrow v}^{2l+2}),\\ &z_{(v)}=a\end{aligned}\bigg\}.
\end{align*}

The minimum weight substructures satisfy the following recurrences.
\begin{proposition}\label{prop:calW_induction}
Let $a\in\{0,1\}$, then
\newline
1) for every $1\leqslant l\leqslant h-1$,
\begin{align*}
\calW_{\sub}^{\min}(\calT_{v\rightarrow C}^{2l+1},a) = \frac{w_{h-l}}{\deg_G(v)}\cdot\lambda_v(a)
+\frac{1}{\deg_G(v)-1}\cdot\sum_{C'\in\calN(v)\setminus\{C\}} \calW_{\sub}^{\min}(\calT_{C'\rightarrow v}^{2(l-1)+2},a).
\end{align*}
\newline 2) for every $0\leqslant l\leqslant h-1$,
\begin{align*}
\calW_\sub^{\min}(\calT_{C\rightarrow v}^{2l+2},a) = \min\bigg\{\sum_{u\in\calN(C)\setminus\{v\}} \calW_{\sub}^{\min}(\calT_{u\rightarrow C}^{2l+1},x_{u})\ \bigg\vert \begin{aligned}x\in\{0,1\}^{\deg_G(C)},\\\lVert x\rVert_1~\mathrm{even}, x_v=a\end{aligned}\bigg\}.
\end{align*}
\end{proposition}

The following claim states an invariant over the messages $\mu_{C\rightarrow v}^{l}(a)$ and $\mu_{v\rightarrow C}^{l}(a)$ that holds during the execution of \NWMS2.
\begin{claim}
Consider an execution of $\NWMS2(\lambda(0),\lambda(1),h,w)$. Then, for every $0\leqslant l \leqslant h-1$,
\begin{eqnarray*}
\mu_{v\rightarrow C}^{(l)}(a) &=& \calW_{\sub}^{\min}(\calT_{v\rightarrow C}^{2l+1},a), \ \ \mathrm{and}\\
\mu_{C\rightarrow v}^{(l)}(a) &=& \calW_\sub^{\min}(\calT_{C\rightarrow v}^{2l+2},a).
\end{eqnarray*}
\end{claim}

\begin{proof}
The proof is by induction on $l$. The induction basis, for $l=0$, holds because $\mu_{v\rightarrow C}^{(0)}(a)=\calW_{\sub}^{\min}(\calT_{v\rightarrow C}^{1},a)=\frac{w_h}{\deg_G(v)}\lambda_v(a)$ for every edge $(v,C)$ of $G$.
The induction step follows directly from the induction hypothesis and Proposition~\ref{prop:calW_induction}.
\end{proof}

%%%%%%%%%%%%%%%%%%%%%%%%%%%%%%%%%%%%%%%%%%%%%%%%%%%%%%%%%%%%%%%%%%%%%%%%
\section{Proof of Lemma~\ref{lemma:ADSbound2}}\label{app:proofofADSbound2}
%%%%%%%%%%%%%%%%%%%%%%%%%%%%%%%%%%%%%%%%%%%%%%%%%%%%%%%%%%%%%%%%%%%%%%%%

\begin{proof}
We prove the lemma by induction on the difference $l-s$.
We first derive an equality for $\E e^{-tY_l}$ and a bound for $\E e^{-tX_l}$.
Since $Y_l$ is the sum of mutually independent variables,
\begin{equation} \label{eqn:1}
\E e^{-tY_l} = \big(\E e^{-t\omega_l \gamma}\big)\big(\E e^{-tX_{l-1}}\big)^{\dL'}.
\end{equation}
By definition of $X_l$ we have the following bound,
\begin{eqnarray*}
 e^{-tX_l} &=& e^{-t\sum_{j=1}^{d'} \min^{[j]}\{Y_l^{(i)}~\vert~ 1 \leqslant i \leqslant \dR'\}}\\
 &=&\prod_{j=1}^{d'}e^{-t\min^{[j]}\{Y_l^{(i)}~\vert~ 1 \leqslant i \leqslant \dR'\}} \\
 &\leqslant& \sum_{\{S\subseteq[\dR']~\vert~|S| = d'\}}\prod_{i\in S}e^{-tY_l^{(i)}}.
\end{eqnarray*}
By linearity of expectation and since $\{Y_l^{(i)}\}_{i=1}^{\dR'}$ are mutually independent variables, we have
\begin{equation}\label{eqn:2}
\E e^{-tX_l} \leqslant \binom{\dR'}{d'}\bigg(\E e^{-tY_l}\bigg)^{d'}.
\end{equation}
By substituting~(\ref{eqn:1}) in (\ref{eqn:2}), we get
\begin{equation}\label{eqn:3}
\E e^{-tX_l} \leqslant \bigg(\E e^{-tX_{l-1}}\bigg)^{(\dL'\cdot d')}\binom{\dR'}{d'} \bigg(\E e^{-t\omega_l\gamma}\bigg)^{d'},
\end{equation}
which proves the induction basis where $s=l-1$.
Suppose, therefore, that the lemma holds for $l-s=i$, we now prove it for $l-(s-1)=i+1$. Then by substituting~(\ref{eqn:3}) in the induction hypothesis, we have
 \begin{align*}
 \E e^{-tX_l} &\leqslant {{\bigg(\E e^{-tX_s} \bigg)}^{(\dL'\cdot d')}}^{l-s}\cdot \prod_{k=0}^{l-s-1}{\bigg(\binom{\dR'}{d'}\big(\E e^{-t\omega_{l-k}\gamma}\big)^{d'}\bigg)^{(\dL'\cdot d')}}^k \\
 &\leqslant {{\bigg[\bigg(\E e^{-tX_{s-1}}\bigg)^{(\dL'\cdot d')}\binom{\dR'}{d'}\! \bigg(\E e^{-t\omega_s\gamma}\bigg)^{d'}\bigg]}^{(\dL'\cdot d')}}^{l-s}
   \!\cdot\!\prod_{k=0}^{l-s-1}\!\!{\bigg(\!\binom{\dR'}{d'}\!\big(\E e^{-t\omega_{l-k}\gamma}\big)^{d'}\bigg)^{\!\!(\dL'\cdot d')}}^k \\
 &= {{\bigg(\E e^{-tX_{s-1}}\bigg)}^{(\dL'\cdot d')}}^{l-s+1}\cdot \prod_{k=0}^{l-s}{\bigg(\binom{\dR'}{d'}\big(\E e^{-t\omega_{l-k}\gamma}\big)^{d'}\bigg)^{(\dL'\cdot d')}}^k,
  \end{align*}
which concludes the correctness of the induction step for a difference of $l-s+1$.
\end{proof}

%%%%%%%%%%%%%%%%%%%%%%%%%%%%%%%%%%%%%%%%%%%%%%%%%%%%%%%%%%%%%%%%%%%%%%%%
\section{Proof of Lemma~\ref{lemma:improvedBound}}\label{app:proofofImprovedBound}
%%%%%%%%%%%%%%%%%%%%%%%%%%%%%%%%%%%%%%%%%%%%%%%%%%%%%%%%%%%%%%%%%%%%%%%%

\begin{proof}
By Lemma~\ref{lemma:ADSbound2}, we have
\begin{align*}
\E e^{-tX_{h-1}}
 \leqslant{{\bigg(\E e^{-tX_s} \bigg)}^{(\dL'\cdot d')}}^{h-s-1}\cdot{\bigg(\binom{\dR'}{d'}\big(\E e^{-t\rho \eta}\big)^{d'}\bigg)^{\frac{(\dL'\cdot d')^{h-s-1}-1}{\dL'\cdot d'-1}}}.
\end{align*}
Note that $\E e^{-t\rho\eta}$ is minimized for $e^{t\rho}=\sqrt{p(1-p)}$. Hence,
\begin{align*}
\E e^{-tX_{h-1}}
\leqslant&{{\bigg(\E e^{-tX_s}\bigg)}^{(\dL'\cdot d')}}^{h-s-1}\cdot {\bigg(\binom{\dR'}{d'}\big(2\sqrt{p(1-p)}\big)^{d'}\bigg)^{\frac{(\dL'\cdot d')^{h-s-1}-1}{\dL'\cdot d'-1}}}\\
\leqslant&
{{\bigg[\bigg(\E e^{-tX_s} \bigg){\bigg(\binom{\dR'}{d'}\big(2\sqrt{p(1-p)}\big)^{d'}\bigg)^{\frac{1}{\dL'\cdot d'-1}}}\bigg]}^{(\dL'\cdot d')}}^{h-s-1}\\
&\cdot {\bigg(\binom{\dR'}{d'}\big(2\sqrt{p(1-p)}\big)^{d'}\bigg)^{-\frac{1}{\dL'\cdot d'-1}}}.
\end{align*}
Let $\alpha \triangleq \min_{t \geqslant 0}\bigg\{\E e^{-tX_s}
\bigg(\binom{\dR'}{d'}\big(2\sqrt{p(1-p)}\big)^{d'}\bigg)^{\frac{1}{\dL'\cdot d'-1}}\bigg\}$. Let $t^* = \argmin_{t \geqslant 0}\E e^{-tX_s}$, then
\begin{align*}
 \E e^{-t^*X_{h-1}}\leqslant \alpha^{(\dL'\cdot d'-1)^{h-s-1}}\cdot{\bigg(\binom{\dR'}{d'}\big(2\sqrt{p(1-p)}\big)^{d'}\bigg)^{-\frac{1}{\dL'\cdot d'-1}}}.
\end{align*}
Using Lemma~\ref{lemma:ADSbound1}, we conclude that
\begin{align*}
\Pi_{\lambda,d,\dL,\dR}(h,\omega^{(\rho)})\leqslant \alpha^{\dL(\dL'\cdot d'-1)^{h-s-1}}\cdot{\bigg(\binom{\dR'}{d'}\big(2\sqrt{p(1-p)}\big)^{d'}\bigg)^{-\frac{\dL}{\dL'\cdot d'-1}}},
\end{align*}
and the lemma follows.
\end{proof}

\end{document}